\documentclass[a4paper,unpublished, onecolumn, 11pt]{quantumarticle}
\pdfoutput=1

\usepackage[utf8]{inputenc}
\usepackage[english]{babel}
\usepackage{amsmath}
\usepackage{amssymb}
\usepackage{stackrel}
\usepackage{enumerate}
\usepackage{multicol}
\usepackage{amsfonts}
\usepackage{amsthm}
\usepackage{graphicx}
\usepackage{stmaryrd}
\usepackage{mathtools}
\usepackage{hyperref}
\usepackage{xspace}
\usepackage[all]{xy}

\usepackage{quiver}

\DeclarePairedDelimiter{\ceil}{\lceil}{\rceil}
\DeclarePairedDelimiter{\floor}{\lfloor}{\rfloor}
\DeclarePairedDelimiter{\inn}{\langle}{\rangle}
\DeclarePairedDelimiter{\norm}{\lVert}{\rVert}

\newcommand{\R}{\mathbb{R}}
\newcommand{\CC}{\mathbb{C}}
\newcommand{\N}{\mathbb{N}}
\newcommand{\cl}{\overline}
\newcommand\Define[1]{\textbf{#1}}

\DeclareMathOperator\Split{Split}
\newcommand{\sa}{\text{sa}}
\newcommand{\opp}{\text{op}}
\newcommand{\im}{\text{im}\,}

\newcommand{\asrt}{\text{asrt}}
\newcommand{\sse}{\subseteq}
\newcommand{\ie}{i.e.\xspace}
\newcommand{\id}{\text{id}}
\newcommand\after{\mathop{\circ}}

\usepackage{scalerel}
\newcommand\ssquare{\scaleobj{0.6}{\square}}

\theoremstyle{definition}
\newtheorem{theorem}{Theorem}
\newtheorem*{theorem*}{Theorem}
\newtheorem{proposition}[theorem]{Proposition}
\newtheorem{lemma}[theorem]{Lemma}
\newtheorem{definition}[theorem]{Definition}
\newtheorem{corollary}[theorem]{Corollary}
\newtheorem{example}[theorem]{Example}
\newtheorem{remark}[theorem]{Remark}

\newcommand{\mult}{\mathrel{\&}}
\newcommand{\commu}{\mathrel{\lvert}}

\usepackage[numbers,sort&compress]{natbib}

\usepackage{tikzit}
\input{zx.tikzdefs}

\tikzstyle{box}=[shape=rectangle, text height=1.5ex, text depth=0.25ex, yshift=0.5mm, fill=white, draw=black, minimum height=5mm, yshift=-0.5mm, minimum width=5mm, font={\small}]
\tikzstyle{Z dot}=[inner sep=0mm, minimum size=2mm, shape=circle, draw=black, fill={rgb,255: red,221; green,255; blue,221}]
\tikzstyle{Z phase dot}=[minimum size=1.2em, font={\footnotesize\boldmath}, shape=rectangle, rounded corners=0.5em, inner sep=0.2em, outer sep=-0.2em, scale=0.8, tikzit shape=circle, draw=black, fill={rgb,255: red,221; green,255; blue,221}, tikzit draw=blue]
\tikzstyle{X dot}=[Z dot, shape=circle, draw=black, fill={rgb,255: red,255; green,136; blue,136}]
\tikzstyle{X phase dot}=[Z phase dot, tikzit shape=circle, tikzit draw=blue, fill={rgb,255: red,255; green,136; blue,136}, font={\footnotesize\boldmath}]
\tikzstyle{hadamard}=[fill=yellow, draw=black, shape=rectangle, inner sep=0.6mm, minimum height=1.5mm, minimum width=1.5mm]
\tikzstyle{vertex}=[inner sep=0mm, minimum size=1mm, shape=circle, draw=black, fill=black]
\tikzstyle{vertex set}=[inner sep=0mm, minimum size=1mm, shape=circle, draw=black, fill=white, font={\footnotesize\boldmath}]
\tikzstyle{target}=[inner sep=0mm, minimum size=3mm, shape=circle, draw=black]

\tikzstyle{hadamard edge}=[-, dashed, dash pattern=on 2pt off 1.5pt, thick, draw={rgb,255: red,68; green,136; blue,255}]
\tikzstyle{brace edge}=[-, tikzit draw=blue, decorate, decoration={brace,amplitude=1mm,raise=-1mm}]
\tikzstyle{diredge}=[->]
\tikzstyle{dashed edge}=[-, dashed, dash pattern=on 2pt off 0.5pt, draw=black]

\usepackage{censor}
\usepackage{nicefrac}
\StopCensoring

\newcommand{\EA}{\mathbf{EA}}

\newcommand{\cat}[1]{\mathbf{#1}}
\newcommand{\catC}{\cat{C}}
\newcommand{\catD}{\cat{D}}
\newcommand{\truth}{\mathbf{1}}
\newcommand{\falsity}{\mathbf{0}}
\newcommand{\pproj}{\mathord{\vartriangleright}}

\newcommand{\Pred}{\mathrm{Pred}}

\title{A computer scientist's reconstruction of quantum theory}

\author{Bas Westerbaan}
\email{bas@westerbaan.name}
\affiliation{Cloudflare}
\thanks{The majority of the work was carried out while employed
    at University College London and Radboud Universiteit.}

\author{John van de Wetering}
\email{john@vandewetering.name}
\affiliation{Radboud Universiteit Nijmegen}
\affiliation{Oxford University}
\orcid{0000-0002-5405-8959}

\begin{document}



\maketitle

\begin{abstract}
The rather unintuitive nature of quantum theory has led numerous people to develop sets of (physically motivated) principles that can be used to derive quantum mechanics from the ground up, in order to better understand where the structure of quantum systems comes from.
From a computer scientist's perspective we would like to study quantum theory in a way that allows interesting transformations and compositions of systems and that also includes infinite-dimensional datatypes. Here we present such a compositional reconstruction of quantum theory that includes infinite-dimensional systems.
This reconstruction is noteworthy for three reasons: it is only one of a few that includes no restrictions on the dimension of a system; it allows for both classical, quantum, and mixed systems; and it makes no a priori reference to the structure of the real (or complex) numbers. This last point is possible because we frame our results in the language of category theory, specifically the categorical framework of effectus theory. 
\end{abstract}

\clearpage
\tableofcontents

\clearpage
\section{Introduction}

Quantum theory is famously unintuitive. Furthermore, it is not a priori clear why its mathematical machinery (complex Hilbert spaces, bounded operators, tensor products) should lead to a correct description of nature.
This has led numerous people throughout the last hundred years to try and reconstruct quantum theory from first principles.
The idea here being that if one can find a set of reasonable assumptions that are only satisfied by quantum theory and not by any other hypothetical physical theory then one has a better grasp on understanding why this mathematics describe nature so well.

There are many such \emph{reconstructions of quantum theory}.
Much early work was based on the orthomodular lattices of von Neumann's \emph{quantum logic}~\cite{birkhoff1936logic}. These approaches focused on the \emph{sharp} observables (projections) of a quantum system and mostly considered an infinite-dimensional system in isolation, i.e.~one which is not composable with other systems (see~\cite{coecke2000operational} for a review).
In contrast, much of the work on reconstructions in the last two decades has instead focused on finite-dimensional systems that can interact with each other and be combined into composite systems. 
Most of these results take an \emph{operational} approach, which entails
that they fundamentally presuppose the nature of classical probability
theory in order to describe classical interactions such as measurement
and probabilistic mixtures of
processes~\cite{barrett2007information,dariano2017firstprinciples}.
This requires the a priori usage of real numbers and convex sets in their frameworks.
The principles themselves in these approaches come in many different guises: some are based on information processing properties~\cite{chiribella2011informational,barnum2014higher,clifton2003characterizing,fivel2012derivation}, others on properties of entanglement~\cite{chiribella2016entanglement,masanes2014entanglement}, properties of pure processes~\cite{tull2016reconstruction,selby2018reconstructing,wetering2018reconstruction}, or on any other of a multitude of properties~\cite{hohn2017quantum,hardy2001quantum,hardy2011reformulating,krumm2017thermodynamics,masanes2011derivation,niestegge2020simple,wetering2018sequential}.
Each of these approaches sheds new light on how quantum theory is `special' among a large selection of hypothetical physical theories.

From a computer scientist's point of view, the compositional nature of many modern works that shed light on the interactions of systems is preferable over the older work that dealt with systems in isolation. However, the restriction to finite dimension is less desirable, as many natural datatypes to describe programs require infinite-dimensional systems
Indeed, in order to describe, say, the natural numbers type in a quantum programming language we require an infinite-dimensional algebra~\cite{cho2016semantics,pechoux2020quantum,cho2016neumann}.
In addition, most reconstructions of quantum theory employ principles that are only satisfied by quantum systems, but not mixed classical-quantum systems. This prevents the inclusion of systems needed to describe quantum programming languages that have a `quantum data/classical control' architecture~\cite{selinger2004towards,green2013quipper}.

One then wonders whether there is a reconstruction of quantum theory that includes infinite-dimensional systems and also allows for mixed classical-quantum systems. In addition it would be desirable if we could sidestep the a priori usage of real numbers and convex sets and instead work in a more abstract categorical setting.
To phrase this question more concretely:
\begin{quote}
 Are there nice assumptions on a category
    such that any such category must be
        a category of quantum types with quantum programs between them?
\end{quote}
In this paper we present
    a reconstruction with these three desirable properties.
Firstly, our assumptions hold for infinite-dimensional types,
in contrast to the axioms of the vast majority of the existing
    reconstructions.
Secondly, our axioms include classical types and mixed quantum-classical types, whereas most reconstructions
    restrict to purely-quantum types.
Finally, our axioms do not presuppose the real or complex numbers.
As far as we are aware, this is the first reconstruction with all these desirable properties together.
Additionally, the core of our reconstruction does not assume a symmetric monoidal structure. We only need the presence of a tensor product for the final step.

Our axioms split roughly into three groups.  The first group
	specifies our basic framework: the category is an \emph{effectus}~\cite{cho2015introduction,kentathesis}, a basic type of structure that has very minimal assumptions while still allowing us to speak of states and predicates.
	We make the additional familiar assumption that the predicate spaces are directed-complete (i.e.~that they form a \emph{dcpo}).

The second group deals with additional categorical structure, \emph{filters} and \emph{comprehensions}~\cite{cho2015quotient}, that imposes the well-behavedness of certain filters and pure maps. 
    We believe these assumptions and the structure they imply might be of interest in their own right, and so we give a name to the sort of effectus with these properties: a \emph{$\diamond$-effectus} (pronounced `diamond-effectus').

Finally, in the third group, we require some more operationally motivated axioms
 which state the well-behavedness of the operation of \emph{sequential measurement} (to wit, we require our predicate spaces to form \emph{sequential effect algebras}~\cite{gudder2002sequential}).

Our assumptions are satisfied by the category of von Neumann algebras with normal positive linear contractions in the opposite direction (representing quantum theory) and also by the category of complete Boolean algebras (representing deterministic classical logic). Our main result is a rough converse to this.
More formally, our reconstruction proceeds in three steps.
First, we show that a category satisfying our assumptions embeds into the product of the category of complete Boolean algebras and the category of directed-complete \emph{JB-algebras}~\cite{hanche1984jordan}. JB-algebras are a type of infinite-dimensional Jordan algebra that are closely related to C$^*$-algebras. 
The reason our category embeds into a product category of Boolean algebras and JB-algebras is because the scalars in the category are `spatial' and can be probabilistic in one part of the space and sharp in another part. By restricting to categories with `irreducible' scalars we show that it either embeds into the category of Boolean algebras \emph{or} the category of \emph{JBW-algebras},
a particularly well-behaved type of JB-algebra that is closely related to von Neumann algebras.
Hence, for irreducible scalars we get a dichotomy between classical deterministic logic, and a theory of quantum systems.
Finally, we impose additional symmetric monoidal structure on the category, so that we can form composite systems.
This forces each JBW-algebra to be a \emph{JW-algebra}, a Jordan algebra that embeds into a von Neumann algebra.

We give a schematic overview of the proof and the different intermediate results and structures in Figure~\ref{fig:flowchart}.

\begin{figure}[!htb]
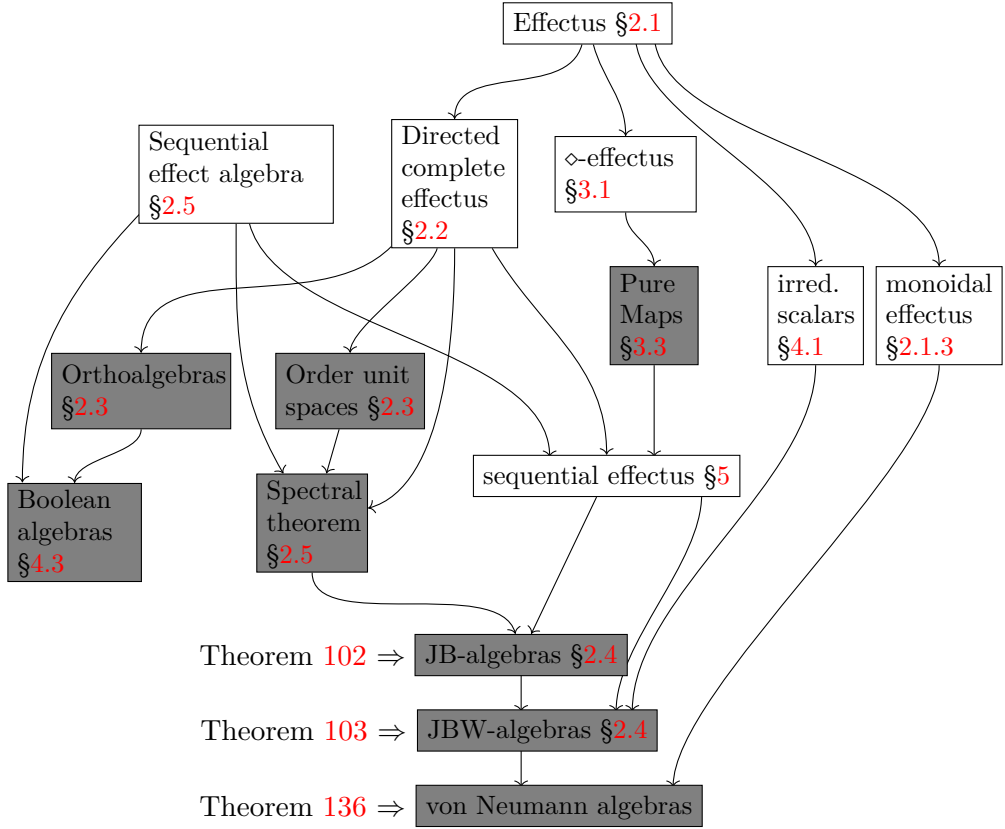

    \ctikzfig{flowchart}
    \caption{Structure of the proof. White boxes represent assumptions, while grey boxes represent derived concepts and structures. An arrow from $A$ to $B$ denote that the concept or proof for $B$ depends on $A$.}
    \label{fig:flowchart}
\end{figure}

While systems in our category correspond to either Boolean algebras or JBW-algebras, our assumptions don't force any of the JBW-algebras to be `quantum-like'. For instance, the category of associative JBW-algebras (or equivalently, commutative von Neumann algebras) and normal positive linear contractions maps satisfies all our assumptions. However, these algebras are all classical in the sense that they correspond to measurable spaces. We see the possibility of fully classical examples as a strength of our approach, as it means our assumptions capture those properties that are shared between classical, quantum, and quantum-classical systems, without restricting to some subset of these systems a priori. The existence of quantum systems can be forced on the category by assuming any of a multitude of assumptions that are only satisfied by quantum systems. For instance, we could assume that each map can be dilated~\cite[Section~3.7.1]{basthesis}, 
similar to the requirement of the existence of purifications in~\cite{chiribella2011informational,tull2016reconstruction,selby2018reconstructing}.

\subsection{Related work}

This reconstruction essentially combines two previous reconstructions by one of the authors~\cite{wetering2018reconstruction,wetering2018sequential}. The first of these~\cite{wetering2018reconstruction} also used the effectus framework and used assumptions related to pure maps. The second~\cite{wetering2018sequential} used assumptions based on sequential measurement. Both of these reconstructions relied on the convex structure imposed by the real numbers and were restricted to finite dimension.
In this paper we combine the assumptions of these reconstructions. This allows us to remove these restrictions on dimension and convexity.

Some other reconstructions that are similar in that they are framed in the language of category theory are that of Tull~\cite{tull2016reconstruction} and Selby et al.~\cite{selby2018reconstructing}. These are both inspired by the Oxford school of categorical quantum mechanics and as such deal with symmetric monoidal categories, dagger structures, and compact closure (i.e.~\emph{cups and cap}s, also known as \emph{map-state duality} or the \emph{Choi-Jamio\l{}kowski isomorphism}). Tull's reconstruction is almost entirely categorical, retrieving a category of matrices over a particular type of ring. To retrieve quantum theory one then only has to impose that the ring in question is the complex numbers. The assumptions of the reconstruction are essentially those of the Pavia reconstruction~\cite{chiribella2011informational}, but then translated into the language of category theory. 
The reconstruction of Selby et al.~\cite{selby2018reconstructing} imposes a more standard GPT framework at the start of the reconstruction, but the assumptions themselves are all clearly motivated from a categorical viewpoint.
These two reconstructions are inherently restricted to finite dimension as compact closure is a core property of them, although it is conceivable that there is a way around that by using non-standard analysis~\cite{EPTCS236.4}. Another significant difference is that these reconstructions rely on the Oxford school of categorical quantum mechanics, whilst our reconstruction is more closely aligned to `standard' category theory in the sense that many of our assumptions can be framed in terms of universal properties.

One selling point of our work is that we don't need to assume the structure of the real numbers a priori. Some other ways to get the correct set of scalars are known as well.
A classical result is that of Sol\`er~\cite{soler1995characterization}, who showed that if an infinite-dimensional generalised Hilbert space over some division ring is orthomodular, then the ring in question must be the real numbers, complex numbers or the quaternions. 
Another approach is given by the work of Heunen~\cite{heunen2009embedding} and Vicary~\cite{vicary2011completeness}. They both derive (related) sets of conditions under which the scalars of a suitable dagger-category embed into the complex numbers, and in Heunen's case, under which the category itself embeds into the category of complex Hilbert spaces. Whereas we work in the setting of effectuses and impose an order-theoretic condition, directed completeness, they work in dagger categories and impose a cardinality condition, that the number of scalars is at most equal to the continuum.
A drawback of their results is that the scalars only \emph{embed} into the complex numbers. For instance, the field of rational numbers is allowed in their results and so is the (non-Archimedean) field of rational functions. This embedding generally does not preserve the ordering of the elements.
Very recently, Heunen and Kornell improved upon the result by Heunen and found a set of categorical conditions that force a category to be equivalent to the category of real or complex Hilbert spaces (containing both finite- and infinite-dimensional spaces)~\cite{heunen2021axioms}. Their result uses Sol\`er's theorem to show the ring of scalars is the field of real or complex numbers.
Their axioms are categorically natural and based on the theory of dagger monoidal categories with dagger biproducts.

While most (modern) reconstructions focus on finite-dimensional systems, there are some exceptions. A particularly relevant one is the work of Alfsen and Shultz~\cite{alfsen2012state,alfsen2012geometry}. They find geometric conditions for when a convex set is isomorphic to the state space of a quantum system or, more generally, a Jordan operator algebra. Our proof works essentially by showing that our spaces satisfy (something similar to) the conditions they find.
A number of reconstructions of infinite-dimensional quantum theory rely heavily on the work of Alfsen and Shultz, for instance~\cite{niestegge2012conditional,landsman1997poisson,guz1981conditional}. Especially this last one resembles our work in that they also assume a completeness condition for the order on predicates, and that they assume the existence of filters, although some other assumptions of~\cite{niestegge2012conditional} do not have a clear motivation.

\subsection{Structure of the paper}
We recall all the definitions and some known results we will need in Section~\ref{sec:prelim}.
In particular, we recall the basic definitions of effectus theory (Section~\ref{sec:effectus}), the notion of directed completeness (Section~\ref{sec:dc-effectus}), order unit spaces (Section~\ref{sec:OUS}), Jordan operator algebras (Section~\ref{sec:jordan-algebra}), and sequential effect algebras (Section~\ref{sec:seqprod}).
Then in Section~\ref{sec:diamond} we will see some consequences of having well-behaved filters and comprehensions, leading to the new definition of a $\diamond$-effectus.
In Section~\ref{sec:decomposition} we show how some of our assumptions, in particular directed completeness, conspire to force an effectus to split into a sharp part and a convex part, which forms the backbone of our reconstruction.
Then in Section~\ref{sec:reconstruction} we present the main results of our reconstruction: that an effectus satisfying our assumptions embeds into the product category of Boolean algebras and JB-algebras. We finish our reconstruction by also considering a tensor product in Section~\ref{sec:monoidal-reconstruction}.
We end the paper with some concluding remarks in Section~\ref{sec:conclusion}.

\section{Preliminaries}\label{sec:prelim}

The assumptions of our reconstruction and the steps in our proof rely on definitions from several somewhat disparate fields, namely effectus theory, sequential effect algebras and Jordan operator algebras. In this section we will recall all these concepts.

\subsection{Effectus theory}\label{sec:effectus}
The basic assumption of our reconstruction is that our category~$\catC$ is an~\emph{effectus}~\cite{cho2015introduction,kentathesis}.
This is a weak structure that allows for a basic notion of state and predicate.
The requirement that a category be an effectus should be compared to the requirement that a set be a topological space: 
one rarely considers just an arbitrary topological space as it has so little structure. The strength of topological spaces
        though, is that they allow for the definition of many
        important notions on top of it.
            Similarly, an effectus on its own has little structure,
                but allows for the definition of many interesting notions.
An effectus can be defined in two ways:
    either axiomatising a category of total maps or of partial maps.
Though we will not use it in the rest of the paper, 
we will give the definition of the total form first as it has the cleanest definition.
Although seemingly obscure at first,
    many categories with a coproduct that behaves as a probabilistic
    disjunction are effectuses.
\begin{definition}
A category~$\catC$ is an \Define{effectus in total form}
iff
\begin{enumerate}
\item $\catC$ has finite coproducts (hence an initial object~$0$)
        and a final object~$1$;
    \item all diagrams of the following
        form\footnote{%
        We write~$\kappa_i$ for coproduct coprojections;
            square brackets~$[f,g]$ for coproduct cotupling;
            $h+k = [\kappa_1 \after h, \kappa_2 \after k]$
            and~$!$ for the unique maps associated to
            either the final object~$1$ or initial object~$0$.} are pullbacks
\begin{equation}
    \vcenter{\xymatrix{
        X+Y \ar[r]^{\id+{!}} \ar[d]_{!+\id} & X+1\ar[d]^{!+\id} \\
    1+Y\ar[r]_{\id+!} & 1+1
}}
    \qquad
    \vcenter{\xymatrix{
    X \ar[r]^{!} \ar[d]_{\kappa_1} & 1 \ar[d]^{\kappa_1} \\
    X+Y\ar[r]_{{!}+{!}} & 1+1
}}
\end{equation}
\item and the following two arrows are jointly monic.
    \begin{equation*}
        \xymatrix@C+2pc  {
            1+1+1  \ar@/^/[r]^{[\kappa_1,\kappa_2,\kappa_2]}
                    \ar@/_/[r]_{[\kappa_2,\kappa_1,\kappa_2]} & 1+1
        }
    \end{equation*}
\end{enumerate}
A \Define{partial map} from $X \to Y$ is an arrow~$X\to Y+1$;
        a \Define{state} on~$X$ is an arrow~$1 \to X$;
        a \Define{predicate} on~$X$ is an arrow~$X \to 1+1$.
The partial maps of an effectus can be composed in the obvious way, and hence we get a category $\text{Par}(\catC)$ of partial maps (formally, $\text{\textunderscore}+1$ is the maybe monad on $\catC$ and $\text{Par}(\catC)$ is its Kleisli category).
\end{definition}

\begin{example}\label{ex:effectus}
    We just give a few examples.  For a more comprehensive
        list, see~\cite{cho2015introduction}.
\begin{enumerate}
\item
    The category of sets and functions is an effectus in total
        form. The states of a set $A$ correspond to the elements
	of $A$ and the predicates correspond to the subsets of $A$.
\item
    The category of sets and probabilistic functions\footnote{I.e.~the Kleisli
            category of the finite distribution monad (but note that the Kleisli category of the Giry monad on measurable spaces is also an example of an effectus).} is an effectus
            in total form.  The states on a set~$A$ correspond to probability distributions on~$A$
                and predicates are maps~$A \to [0,1]$.
\item
	The opposite category of (finite-dimensional) C$^*$-algebras
        with positive unital linear maps forms an effectus in total form.
        States of an
	algebra $\mathfrak{A}$ are positive unital linear maps
	$\omega:\mathfrak{A}\rightarrow \CC$, and the predicates
	correspond to elements of $[0,1]_{\mathfrak{A}}$.
\end{enumerate}
\end{example}

\begin{remark}
    In this last example we used the opposite category of C$^*$-algebras. This is because C$^*$-algebras are spaces of observables (and hence predicates), while effectuses are defined in terms of states. Using the language of physicists we would say that effectuses are in the Schr\"odinger picture, while C$^*$-algebras are in the Heisenberg picture.
    In this paper we will often see the necessity of working with an opposite category for this reason.
\end{remark}

For our purposes it will be more convenient to work with the category
    of partial maps of an effectus.
That category can be axiomatised on its own
    as an \emph{effectus in partial form},
    but that requires some preparation.
To start, that category comes with a partial addition on the maps.

\begin{definition}
\label{def:pcm}
A \Define{partial commutative monoid}
(\Define{PCM})
is a set $M$ with an element~$0\in X$ 
and a \emph{partial} binary
        operation~$\ovee\colon M\times M\rightharpoonup M$
such that for all $x,y,z \in M$
\begin{itemize}
	\item $(x\ovee y)\ovee z = x\ovee(y\ovee z)$ (associativity),
    \item $x\ovee y= y\ovee x$ (commutativity), \emph{and}
	\item $0\ovee x = x$ (unitality).
\end{itemize}
Here `$=$' is taken to be a \emph{Kleene equality}.\footnote{
Kleene equality: if either side is defined, then so is the other,
    and they are equal.
Hence an equation like $x\ovee y = z$ is taken to
mean both that $x\ovee y$ is defined, as well 
    as that we have the equality $x\ovee y = z$.}
We write~$x\perp y$ to denote $x\ovee y$ is defined.
A function $f:M\rightarrow N$
    between PCMs
    is \Define{additive} if $f(0)=0$ and
$f(x)\ovee f(y) = f(x\ovee y)$ for all $x\perp y$ in $M$.
The Cartesian product~$M\times N$ of two PCMs is again
    a PCM in the obvious way.  A map~$g\colon M\times N\to L$
        is \Define{biadditive}
            if its restrictions
                $g(x,-)$, $g(-,y)$ for arbitrary~$x$ and~$y$ are additive.
We say a category is \Define{enriched over PCMs} if each homset is a PCM and the composition maps are biadditive.
\end{definition}

A category enriched over PCMs has a partial addition operation defined on its morphisms that interacts suitably with composition. This acts as an abstraction and generalisation of the \emph{coarse-graining} operation present in, for instance, generalised probabilistic theories~\cite{barrett2007information}. When the sum of two morphisms $f$ and $g$ is defined, it means that there is a sense in which $f$ and $g$ can coexist as different branches of a probabilistic process. The sum morphism $f\ovee g$ then corresponds to their coarse-graining where we forget which of the two processes actually happened.
In an effectus we also have coproducts to model the probabilistic disjunction of systems, and the coarse-graining operation interacts suitably with these coproducts.

\begin{definition}
A category $\catC$ with zero morphisms $0\colon A\to B$
(such as when it is enriched over PCMs) has for
each coproduct $\coprod_{j\in J} A_j$
 \Define{partial projections}
$\pproj_i\colon \coprod_{j\in J} A_j\to A_i$
characterized by $\pproj_i\circ \kappa_i=\id$
and $\pproj_i\circ \kappa_k=0$ for $k\neq i$.
A family $(f_j\colon B\to A_j)_{j\in J}$ of morphisms
in $\catC$ is \Define{compatible} if there exists an
$f\colon B\to \coprod_{j\in J} A_j$ such that
$\pproj_j\circ f= f_j$ for each $j\in J$.

A \Define{finitely partially additive category} (finPAC)~\cite{cho2015total}
is a category with finite coproducts that is enriched over PCMs 
    so that the coproduct and PCM operations interact suitably:
\begin{itemize}
\item \Define{Compatible sum axiom}:
Compatible pairs of morphisms $f,g\colon A\to B$
are summable in $\catC(A,B)$.
\item \Define{Untying axiom}:
If $f,g\colon A\to B$ are summable,
then $\kappa_1\circ f,\kappa_2\circ g\colon A\to B+B$
are summable too.
\end{itemize}
\end{definition}

An effectus in partial form is a finPAC, but we need a bit more:
    that the predicates, the homsets $\catC(A,I)$, are a special kind of PCM.

\begin{definition}
\label{def:EA}
    An \Define{effect algebra}~\cite{foulis1994effect} is a PCM $(E,\ovee, 0)$
with a `top' element $1\in E$ such that for each $x\in E$,
    \begin{itemize}
        \item there is a unique $x^\bot\in E$
    (called the orthosupplement) satisfying $x\ovee x^\bot=1$, and
    \item $x\perp 1$ implies $x=0$.
    \end{itemize}
For~$x,y \in E$ we write~$x \leq y$ whenever there is a~$z \in E$
    with~$x \ovee z = y$.
This turns~$E$ into a poset with minimum~$0$ and maximum~$1=0^\perp$.
The map~$x \mapsto x^\perp$ is an order anti-isomorphism.
Furthermore~$x \perp y$ if and only if~$x \leq y^\perp$.
We write $\EA$ for the category of
effect algebras and additive maps.
Note that additive maps automatically preserve the order (i.e.~are monotone).
\end{definition}

\begin{example}\label{ex:orthomodularlattice}
    Let $(B,0,1,\wedge, \vee, (\ )^\perp)$ be an orthomodular lattice. Then $B$ is an effect algebra with the partial addition defined by 
    $x\perp y \iff x\wedge y = 0$ and in that case  $x\ovee y = x\vee y$.
    The orthosupplement $(\ )^\perp$ is given by the orthocomplement itself.
    The lattice order coincides with the effect algebra order (defined
    above). See e.g.~\cite[Prop.~27]{basmaster}.
\end{example}

\begin{example}
	For a unital C$^*$-algebra $\mathfrak{A}$, the set of effects $[0,1]_{\mathfrak{A}}$ is an effect algebra. This is the motivating example.
\end{example}

\begin{definition}
\label{def:effectus}
    An \Define{effectus in partial form} is a finPAC $\catC$
    with a distinguished \Define{unit} object $I\in\catC$
satisfying
the following conditions.
\begin{itemize}
\item
The PCM $\catC(A,I)$ is an effect algebra for all~$A$.
We write $\truth_A$
and~$\falsity_A=0_{A I}$
for the top and bottom of $\catC(A,I)$.
\item
$\truth_B\circ f=\falsity_A$ implies $f=0_{A B}$
for all $f\colon A\to B$.
\item
$\truth_B\circ f\perp \truth_B\circ g$
implies $f\perp g$
for all $f,g\colon A\to B$.
\end{itemize}
    We call a map~$f: A\to B$ \Define{total} when~$\truth_B \circ f = \truth_A$.
\end{definition}

Viewing an effectus (in partial form) as an abstraction of a generalised probabilistic theory, we can give an interpretation to these axioms. That the predicates form an effect algebra means, first, that we have a \emph{deterministic} predicate $\truth_A$ for every system $A$ so that the processes in the theory are non-signalling~\cite{coecke2014terminality}, and second, that for every predicate~$p$ we have its \emph{negation}~$p^\perp$. The existence of negations in non-signalling GPTs is usually a consequence of the ability to coarse-grain measurements.
The second and third axioms can be interpreted as a weak form of \emph{operational equivalence}, stating that maps are zero, respectively summable, when they are zero, respectively summable, on every predicate~\cite{tull2016}.

We can freely switch between effectuses in total and partial
    form:
\begin{remark}
	Let $\catC$ be an effectus in total form. Then $\text{Par}(\catC)$
	is an effectus in partial form. Conversely, for an effectus
	in partial form $\catD$, the category of \Define{total maps}
	$\text{Tot}(\catD)$ is an effectus in total form.
    This is, in fact, a 2-categorical equivalence between the category of
        effectuses in total form and the category of effectuses in partial form~\cite{kentathesis}.
\end{remark}

\begin{example}
	Adapting Example~\ref{ex:effectus} to the partial case we see that the category of sets and partial functions is an effectus in partial form. So is the Kleisli category of the subdistribution monad and the opposite category of C$^*$-algebras with contractive positive linear maps.
	The category $\textbf{EA}^\opp$ is also an effectus in partial form.
\end{example}

For the remainder of the paper we will work solely with effectuses in
    partial form and simply refer to them as effectuses.
For clarity, let us translate some of the important notions:
    in an effectus (in partial form),
    \begin{itemize}
    \item a \Define{predicate} is a map~$A \to I$,
    \item a \Define{state} is a total map~$I \to A$, and
    \item a \Define{scalar} is a map~$I \to I$.
    \end{itemize}

\begin{definition}
    For any object~$A$ in an effectus we write~$\Pred(A)$ for the effect algebra of predicates
    on~$A$. For a morphism $f\colon A\rightarrow B$ we write $\Pred(f)\colon\Pred(B)\rightarrow \Pred(A)$ for the map defined by $\Pred(f)(p) := p\circ f$.
\end{definition}
It is clear that $\Pred$ is a functor from $\catC$ to $\textbf{EA}^\opp$.
        The image~$\Pred(\catC)$ is an effectus,
            and it is equivalent to~$\catC$
                iff~$\Pred$ is faithful, which is equivalent to the following.
\begin{definition}[{cf.~\cite{sigma}}]
    We say an effectus $\catC$ is \Define{separated by predicates}
    if for a pair of morphisms $f,g\colon A\rightarrow B$ we have $f=g$
    when~$p\circ f=p\circ g$ for all $p\in \Pred(B)$.
\end{definition}
Separation by predicates is analogous to the condition of \emph{local tomography} in the setting of generalised probabilistic theories~\cite{barrett2007information}. We also have a dual definition, which asks the same, but for states.

\begin{definition}
    Let $\catC$ be an effectus. We say it is \Define{separated by states} when for all pairs of morphisms $f,g\colon A\rightarrow B$ we have $f=g$ iff $f\circ \omega = g\circ \omega$ for all states $\omega\colon I\rightarrow A$.
\end{definition}

Just as with predicates, we can also construct a `state functor', which goes into a category of \emph{abstract convex sets}~\cite{basthesis,jacobs2015states,sigma}, but we will not need this in this paper.

\subsubsection{Effect monoids}

The set of scalars~$\catC(I,I)$ in an effectus has a rich structure:
    as a set of predicates on~$I$ it's an effect algebra,
    but it also has a multiplication that comes from
    the composition of scalars.  Its structure is axiomatised as follows.

\begin{definition}\label{def:effectmonoid}
    An \Define{effect monoid}\footnote{The
    category of effect algebras has an algebraic
    tensor product that makes the category symmetric
    monoidal~\cite{jacobs2012coreflections}. The monoids in the
    category of effect algebras resulting from this tensor product
    are the effect monoids, hence the name.}
    \cite{jacobs2011probabilities}
    is an effect algebra $(M,\ovee, 0, ^\perp,
    \,\cdot\,)$ with an additional (total) binary operation $\,\cdot\,$,
such that the following conditions hold 
    for all~$a,b,c \in M$.
\begin{itemize}
\item Unit: $a\cdot 1 = a = 1\cdot a$.
\item Distributivity: if $b\perp c$, then $a\cdot b\perp  a\cdot c$,\quad
    $b\cdot a\perp c\cdot a$,
        $$a\cdot (b\ovee c) \ =\  (a\cdot b) \ovee (a\cdot c),\quad \text{and}
        \quad (b\ovee c)\cdot a \ =\  (b\cdot a) \ovee (c\cdot a).$$
Or, in other words:
    the operation~$\,\cdot\,$ is bi-additive.
\item Associativity: $a \cdot (b\cdot c) = (a \cdot b) \cdot c$.
\end{itemize}
 We call an element $p$ of~$M$ \Define{idempotent}
 whenever~$p^2 := p \cdot p = p$.
\end{definition}


\begin{example}\label{ex:booleanalgebra}
Any Boolean algebra  $(B,0,1,\wedge,\vee,(\ )^\perp)$,
    being an orthomodular lattice,
    is an effect algebra by Example~\ref{ex:orthomodularlattice},
and, moreover,  a commutative effect monoid with
    multiplication defined by $x \cdot y= x \wedge y$.
\end{example}

\begin{example}
   In any effectus, the set of scalars is an effect monoid
        with~$s \cdot  t := s \circ t$.
\end{example}

\begin{example}\label{ex:CX}
    Let $X$ be a compact Hausdorff space and denote its space of
    continuous functions into the complex numbers by
    $C(X):=\{f\colon X\rightarrow \CC, f\text{ continuous}\}$. This is
    a commutative unital C$^*$-algebra
    (and conversely by the Gel'fand theorem, any
    commutative C$^*$-algebra with unit is of this form). Its
    unit interval $[0,1]_{C(X)} = C(X,[0,1])$ consisting of continuous functions $f\colon X\rightarrow [0,1]$ is a
    commutative effect monoid.

\end{example}

\begin{remark}
    A physical or logical theory which has probabilities of the form $[0,1]_{C(X)}$ can be seen as a theory with a natural notion of space, where probabilities are allowed to vary continuously over the space $X$. This is explored
    in for instance Ref.~\cite{moliner2017space}.
\end{remark}

\begin{example}
    Given two effect algebras/monoids~$E_1$ and $E_2$ we define
    their \Define{direct sum} $E_1\oplus E_2$ as the Cartesian
    product with pointwise operations. This is again an effect
    algebra/monoid. Effect algebras/monoids that cannot be written as
    a non-trivial direct sum we call \Define{irreducible}.
\end{example}

\begin{example}\label{cornersexample}
    Let~$M$ be an effect monoid and let~$p \in M$ be some idempotent.
    Define $pM := \{p \cdot a; \ a\in M\}$. This
        is an effect monoid with~$(p\cdot a)^\perp := p \cdot a^\perp$
            and all other operations inherited from~$M$.
    The map~$a \mapsto (p\cdot a, p^\perp \cdot a)$
        is an isomorphism~$M \cong p M \oplus p^\perp M$~\cite{first}.
In particular, an effect monoid is irreducible iff it has no non-trivial idempotents.
\end{example}

\subsubsection{Filters and comprehensions}\label{sec:quotient-comprehension}

So far we have discussed the general structure of an effectus, which is present in a large class of examples. Now we will look at additional structure that is more specialised.

We will require the existence of certain universal maps
    into and out of subsystems,
    which can be motivated operationally~\cite{wetering2018reconstruction} as
        filters and arise categorically as adjunctions~\cite{cho2015quotient}
        (cf.~Remark~\ref{chainofadjs}). Additionally, filters and comprehensions are closely related to the categorical notion of (co)kernels (cf.~Remark~\ref{rem:kernels}) and hence to the notion of \emph{ideal compressions} of~\cite{chiribella2011informational} (see~\cite[Section~4.4.2]{tull2019phdthesis} for details).

\begin{definition}\label{def:comprehension}
    Let $p\colon A\rightarrow I$ be a predicate in an effectus. A
    \Define{comprehension} for $p$ consists of an object $A_p$ and
    a map $\pi_p\colon A_p\rightarrow A$ such that $\truth_A\circ \pi_p
    = p\circ \pi_p$ that is \Define{final} with this property:
    whenever $f\colon B\rightarrow A$ is such that $\truth_A\circ f =
    p\circ f$ then there is a unique $\bar{f}\colon B\rightarrow A_p$
    with~$\pi_p \circ \bar{f} = f$, that is:
    \[\begin{tikzcd}[ampersand replacement = \&]
    A_p \arrow{r}{\pi_p}\& A \\
    B \arrow[dotted]{u}{\bar{f}}\arrow{ru}[swap]{f}\&
    \end{tikzcd}\] 
    We say an effectus \Define{has comprehensions} when every predicate has a comprehension.
\end{definition}

\begin{definition}\label{def:filter}
    Let $p\colon A\rightarrow I$ be a predicate in an effectus. A \Define{filter}\footnote{%
        A filter for~$p$
            is exactly the same thing as what is called a \Define{quotient}
                for~$p^\perp$ in many other papers on effectuses~\cite{cho2015introduction}.
        In those papers~$\xi_p$ correspond to our~$\xi^{p^\perp}$.}
    for $p$ is an object $A^p$ and map $\xi^p\colon A\rightarrow A^p$
    such that $\truth\circ\xi^p\leq p$ which is \Define{initial}
    for this property: for any map $f\colon A\rightarrow B$ which satisfies
    $\truth\circ f\leq p$ there is a unique $\bar{f}\colon A^p\rightarrow
    B$ with~$\bar{f} \circ \xi^p = f$, that is:
    \[\begin{tikzcd}[ampersand replacement = \&]
    A^p \arrow[dotted,swap]{d}{\bar{f}}\&\arrow[swap]{l}{\xi^p} A\arrow{dl}{f} \\
     B\&
    \end{tikzcd}\]
    We say an effectus \Define{has filters} when every predicate has a filter.
\end{definition}

The reason we call these maps filters is because applying a filter $\xi^p$ corresponds in our categories of interest to the `post-selection' of the predicate $p$, so that after application we have `filtered' the state to ensure $p$ is true.

Note that as filters and comprehensions are defined by a universal properties, that they are unique up to unique isomorphism.

\begin{example}\label{ex:quotient-comprehension}
In~\cite{cho2015quotient} many examples of categories
    with filters and comprehension are given.
Here we will restrict ourselves
    to discussing them for the `quantum' example
        of the C$^*$-algebra~$B(\mathcal{H})$ of bounded operators on a Hilbert space
        with positive linear contractions in opposite direction between them.
    Let $p\in B(\mathcal{H})$ be an effect, i.e.~$0\leq p \leq 1$. Denote
    by~$P$ the largest projection (idempotent effect)
        below~$p$, \ie~$P$ projects to the eigenspace
            of~$p$ of eigenvalue~$1$. Denote this space by~$\mathcal{K}\sse \mathcal{H}$. Then the \Define{standard comprehension}
    of~$p$ is the map $\pi_p\colon B(\mathcal{H})\rightarrow B(\mathcal{K})$ given by~$\pi_p(B) = PBP$.
    Now let~$\mathcal{K}'\sse \mathcal{H}$ be~$\mathcal{K}'=(\ker
    p)^\perp$, \ie~the closure of the eigenspaces of $p$ of
    non-zero eigenvalue. Then $p$'s \Define{standard filter} is the map $\xi^p\colon  B(\mathcal{K}')\rightarrow
    B(\mathcal{H})$ given by $\xi^p(q) = \sqrt{p}q\sqrt{p}$.
\end{example}

\begin{remark}\label{rem:kernels}
    It can be shown that an effectus has comprehensions iff it
    has \emph{kernels}~\cite[\S 200]{basthesis}. An effectus has \emph{cokernels} iff
    all maps have an \emph{image} and every \emph{sharp} effect
    has a filter~\cite[\S 205]{basthesis}. We will give definitions of the image and
    sharpness later (in Section~\ref{sec:diamond}), but for now let us note that we can hence interpret 
    filters as `fuzzy cokernels'.
\end{remark}

Filters and comprehensions have a different categorical characterisation
    due to Jacobs.
\begin{remark}\label{chainofadjs}
    Let $\catC$ be an effectus. Let $\Pred_{\ssquare}(\catC)$ denote its \Define{Grothendieck category} which has as objects pairs $(A\in\catC, p\in \Pred(A))$ and morphisms $f\colon (A,p)\rightarrow (B,q)$ given by $f\colon A\rightarrow B$ satisfying $p\leq (q^\perp\circ f)^\perp$.
    There is an obvious forgetful functor $U\colon \Pred_{\ssquare}(\catC)\rightarrow \catC$. Conversely there are two canonical ways to embed $\catC$ into $\Pred_{\ssquare}(\catC)$, namely by mapping an object $A$ to $(A,\falsity)$ and by mapping $A$ to $(A,\truth)$. These two embeddings turn out to be left and right adjoint to the forgetful functor~\cite{cho2015quotient}:
\[\begin{tikzcd}
    {\Pred_{\ssquare}(\catC)} \\
    \\
    {\catC}
    \arrow["{\falsity}"{name=0, description}, from=3-1, to=1-1, shift left=0, curve={height=-30pt}]
    \arrow["{\truth}"{name=1, description}, from=3-1, to=1-1, shift right=0, curve={height=30pt}]
    \arrow["{\xi}"{name=2, description}, from=1-1, to=3-1, curve={height=80pt}]
    \arrow["{\pi}"{name=3, description}, from=1-1, to=3-1, curve={height=-80pt}]
    \arrow["{U}"{name=4, description}, from=1-1, to=3-1]
    \arrow["\dashv"{rotate=0}, from=0, to=4, phantom]
    \arrow["\dashv"{rotate=0}, from=4, to=1, phantom]
    \arrow["\dashv"{rotate=0}, from=1, to=3, phantom]
    \arrow["\dashv"{rotate=0}, from=2, to=0, phantom]
\end{tikzcd}\]
    The $\falsity$ embedding has a left adjoint
    iff $\catC$ has filters\footnote{%
        This is the reason that filters for~$p$ are referred to
            as \emph{quotients} for~$p^\perp$ in the effectus literature.}, and the $\truth$ embedding has a right adjoint iff $\catC$ has comprehensions~\cite[Chapter~5]{kentathesis}.
\end{remark}

We note a number of properties of filters and comprehension that we will use without further reference.
\begin{proposition}[{\cite{cho2015introduction}}]
    Let $\catC$ be an effectus which has filters and comprehensions.
    \begin{itemize}
        \item Every filter is epic, every comprehension is monic.
        \item If $\xi$ is a filter for $a$, then $\truth\circ\xi = a$.
        \item Comprehensions are total: $\truth\circ\pi = \truth$.
    \end{itemize}
\end{proposition}

\subsubsection{Monoidal effectuses}\label{sec:monoidal-effectus}

In most works dealing with GPTs, the notion of a composite system is important. To talk about composite systems in a category we need monoidal structure: a tensor product.
Effectuses don't need to have monoidal structure, and our main result does also not require the existence of a tensor product. However, to make the final jump in our reconstruction from general JBW-algebras to von Neumann algebras, we will require a tensor product. So let us give a definition of a monoidal effectus.

\begin{definition}\label{def:monoidal-effectus}
    We say an effectus is \Define{monoidal} when it has a symmetric monoidal structure $(\otimes, I)$ such that 
    \begin{itemize}
        \item the tensor unit $I$ is also the designated unit object of the effectus,
        \item the tensor product is `biadditive', i.e.~for any morphisms~$f,g,h$ with~$f\perp g$ we have~$(f\ovee g)\otimes h = (f\otimes h)\ovee (g\otimes h)$ and $0\otimes h = 0$,
        \item and the tensor product preserves $\truth$ --- that is: $\truth_A\otimes\truth_B = \truth_{A\otimes B}$.
    \end{itemize}
\end{definition}

Let $\lambda_A\colon I\otimes A\rightarrow A$ denote the natural isomorphism for the tensor unit and let~$s,t\colon I\rightarrow I$ be some scalars. Then for any morphism $f\colon A\rightarrow B$ we can define the map $s\cdot f$ as the composition $s\cdot f := \lambda_B\circ(s\otimes f)\circ\lambda_A^{-1}$. This gives us a scalar multiplication on morphisms in a monoidal effectus. 
Let us note the following straightforwardly verifiable facts.
\begin{lemma}\label{lem:monoidal-effectus-facts}
    Let $\catC$ be a monoidal effectus, and let $s,t\colon I\rightarrow I$ be scalars. Then the following holds.
    \begin{itemize}
        \item Scalar multiplication respects composition: for any $f\colon A\rightarrow B$ and $g\colon B\rightarrow C$ we have $g\circ (s\cdot f) = s\cdot (g\circ f) = (s\cdot g)\circ f$.
        \item Scalar multiplication respects addition: for any $f \perp g\colon A\rightarrow B$ we have $(s\ovee t)\cdot f = s\cdot f \ovee t\cdot f$ and $s\cdot (f\ovee g) = s\cdot f \ovee s\cdot g$.
        \item For any predicate $p\colon A\rightarrow I$ we have $s\cdot p = s\circ p$. In particular $s\cdot t = s\circ t$ so that $s\cdot (t\cdot f) = (s\cdot t)\cdot f = (s\circ t)\cdot f$.
    \end{itemize}
\end{lemma}

\subsection{Directed completeness}\label{sec:dc-effectus}
We will require
    the predicates to form a \emph{dcpo}: a directed-complete poset.
This requirement turns out to be surprisingly strong.

\begin{definition}
	We say an effectus~$\catC$ is \Define{directed complete} when all
    the the predicate spaces $\Pred(A)$ are directed complete.\footnote{%
        An effect algebra~$E$ is said to be directed complete,
            when every upwards-directed subset~$U \subseteq E$
            (i.e.~where for every~$x,y\in U$ there exists $z\geq x,y$ in $U$)
            has a supremum. 
            As~$(\ )^\perp$ is an order anti-automorphism
                this upwards-directed completeness is equivalent to
                downwards-directed completeness.}
    If in addition these suprema are preserved
        by all maps (i.e.~the maps are \emph{Scott continuous})
    then we say~$\catC$ is \Define{normal}.
\end{definition}

\begin{example}
	The category of sets is a  normal effectus, as the predicate spaces are all complete Boolean algebras.
	The category of finite-dimensional C$^*$-algebras and positive linear contractions
        (in the opposite direction) is normal.
    This is not the case when including infinite-dimensional algebras.
    However, the category of von Neumann algebras with (normal)
        positive linear contractions is a (normal) directed-complete effectus.\footnote{%
            In fact, a C$^*$-algebra is a von Neumann algebra
                iff its unit interval is directed-complete
                and it is separated by its normal states.}
\end{example}

The scalars in a directed-complete effectus
    form a directed-complete effect monoid.
    In contrast to arbitrary effect monoids,
        the directed-complete ones are well-understood.

\begin{example}
    Any complete Boolean algebra is a directed-complete effect monoid.
\end{example}

\begin{example}
    Let $X$ be an \Define{extremally-disconnected} compact Hausdorff
    space, i.e.~where the closure of every open set is  open.
    Then $C(X, [0,1])$ is a directed-complete effect monoid.
\end{example}

\begin{theorem}[\cite{first}]\label{thm:effect-monoids}
	Let $M$ be a directed-complete effect monoid. Then there
	exists a complete Boolean algebra $B$ and an
    extremally-disconnected compact Hausdorff space $X$ such that
    $M\cong B\oplus C(X, [0,1])$.
\end{theorem}

\begin{corollary}\label{corscalars}
   If~$\catC$ is a directed-complete effectus,
        then there is an extremally-disconnected compact Hausdorff spaces~$X$
        and a complete Boolean algebra~$B$
        with~$\Pred (I) \cong B \oplus C(X,[0,1])$.
\end{corollary}

The characterisation result of Theorem~\ref{thm:effect-monoids} has a corollary for irreducible effect monoids, also proven in~\cite{first}.
\begin{theorem}
    Let $M$ be an irreducible directed-complete effect monoid. Then $M$ is isomorphic (as an effect monoid) to $\{0\}$, $\{0,1\}$ or $[0,1]$.
\end{theorem}

Hence, in a directed-complete effectus with irreducible scalars we have three possibilities for the scalars. These three different possibilities were analysed in~\cite{sigma}.
If $\Pred(I)\cong\{0\}$ the entire category is equivalent to the trivial one-object category, so we can safely ignore this possibility.
If $\Pred(I)\cong \{0,1\}$, then we don't have any immediate useful consequences, however if we assume the effectus is separated by states or separated by predicates, then this implies a lot of structure, namely that all the predicate spaces are \emph{orthoalgebras}. We will look at these in the next section, together with their counterpart, \emph{convex} effect algebras that arise when $\Pred(I)\cong[0,1]$.

\subsection{Orthoalgebras, convexity and order unit spaces}\label{sec:OUS}

\begin{definition}\label{def:orthoalgebra}
    Let $E$ be an effect algebra. It is an \Define{orthoalgebra}
        when~$0$ is the only self-summable element;
        i.e.~when for every~$a$ with~$a \perp a$, we have~$a = 0$.
    We denote the full subcategory of~$\textbf{EA}$ consisting of the
    orthoalgebras by~$\textbf{OA}$.
\end{definition}
Examples of orthoalgebras include orthomodular lattices and Boolean
    algebras.

\begin{proposition}[{\cite{sigma}}]
    Let $\catC$ be an effectus separated by states and where $\Pred(I)\cong \{0,1\}$. Then every predicate space is an orthoalgebra.
\end{proposition}

\begin{definition}\label{def:convex}
    An effect algebra is \Define{convex}~\cite{gudder1999convex} 
      when there is
   a map $\cdot :  [0,1]\times E\rightarrow E$, 
   where $[0,1]$ is the regular unit interval, 
   obeying
    the following axioms for all $x,y\in E$ and~$\lambda, \mu \in
    [0,1]$:
    \begin{itemize}
        \item $\lambda\cdot (\mu\cdot x) = (\lambda\mu)\cdot x$.
    \item If $\lambda+\mu \leq 1$, then $\lambda \cdot x \perp
    \mu \cdot x$ and $\lambda \cdot x \ovee \mu \cdot x =
    (\lambda+\mu)\cdot x$.
        \item $1\cdot x = x$.
        
        \item $\lambda\cdot (x\ovee y) \ = \ \lambda\cdot x \ovee \lambda\cdot y$.
    \end{itemize}
    Denote by $\textbf{EA}_c$, respectively $\textbf{DCEA}_c$, the subcategory of $\textbf{EA}$ consisting of (directed-complete) convex effect algebras and morphisms that preserve the convex action.
\end{definition}

\begin{example}
    Let $V$ be an ordered real vector space (such as the space of self-adjoint elements of a C$^*$-algebra). Then any interval $[0,u]_V$ where $u\geq 0$ is a convex effect algebra with the obvious action of the real unit interval. 
    Conversely, for any convex effect algebra~$E$, we can find an ordered real vector space $V$ and $u\in V$ such that $E$ is isomorphic as a convex effect algebra to $[0,u]_V$~\cite{gudder1998representation}.
This is, in fact, an equivalance of categories~\cite{jacobs2016expectation}.
\end{example}

Ordered real vector spaces
    play a central role in \emph{generalised probabilistic
    theories}~\cite{barrett2007information} that are often
        used in operational reconstructions of quantum
    theory, cf.~\cite{tull2016,gogioso2018categorical}.
For directed-complete convex effect algebras, this equivalence
    restricts to a more specific type of vector space.

\begin{definition}
    An \Define{order unit space} (OUS) $(V,\leq,1)$ is an ordered vector space $(V,\leq)$ with a designated \Define{order unit} $1$ such that the induced semi-norm defined by $\norm{v} := \inf\{\lambda\in \R~;~-\lambda 1 \leq v \leq \lambda 1\}$ is a norm and where the positive cone of~$V$ is closed in its topology. A \Define{Banach} OUS is furthermore complete in its norm.
    We say an OUS is \Define{directed complete} when its unit interval is.%
    \footnote{This is equivalent to the whole OUS being \emph{bounded} directed complete, i.e.~where every bounded directed subset has a supremum.}
    Denote by \textbf{DCOUS} the category of directed-complete order unit spaces and positive linear contractions. Note that any directed-complete OUS is Banach~\cite[Lemma~1.1]{wright1972measures}.
\end{definition}

\begin{proposition}[{cf.~\cite[Prop.~55]{sigma}}]\label{prop:convex-OUS-equiv}
The equivalence between ordered vector spaces and convex effect algebras restricts to an equivalence $\textbf{DCOUS} \cong \textbf{DCEA}_c$.
\end{proposition}

\begin{proposition}[{\cite{sigma}}]
    Let $\catC$ be an effectus with $\Pred(I)\cong [0,1]$. Then all predicate spaces are convex effect algebras with the convex action given by $\lambda\cdot p := \lambda\circ p$ for $\lambda\in\Pred(I)$.
\end{proposition}

This last result in particular implies that if $\catC$ is additionally directed complete that all the predicate spaces are the unit intervals of directed-complete order unit spaces, and thus that the predicate functor gives a functor from $\catC$ to $\textbf{DCOUS}^\opp$.

\subsection{Jordan operator algebras}\label{sec:jordan-algebra}

Our reconstruction of quantum theory will show that our category embeds into a category of \emph{Jordan algebras}. These are a type of algebras originally introduced as a generalisation of a quantum system~\cite{jordan1993algebraic}, but were quickly found to be very close to regular quantum systems. Indeed, the type of infinite-dimensional Jordan algebras we consider here, JBW-algebras, can be shown to embed into a von Neumann algebra up to a so-called `exceptional ideal', and hence we don't lose much by working with Jordan algebras instead of C$^*$-algebras.

In this section we introduce the necessary concepts related to JBW-algebras.
First, let us introduce the `Jordan version' of a C$^*$-algebra.

\begin{definition}[{\cite[Proposition~3.1.6]{hanche1984jordan}}]
    \label{def:JB-algebra}
    A \Define{JB-algebra} $(A,*,1,\leq)$ is a Banach order unit space equipped with a binary operation $*\colon A\times A\rightarrow A$ satisfying for all $a,b,c\in A$:
    \begin{itemize}
        \item Commutativity: $a*b=b*a$.
        \item Unit: $a*1 = 1*a= a$.
        \item The \Define{Jordan identity}: $(a*b)*(a*a) = a*(b*(a*a))$.
        \item If $-1\leq a \leq 1$, then $0\leq a*a \leq 1$.
    \end{itemize}
\end{definition}

\begin{example}
    Let $\mathfrak{A}$ be a unital C$^*$-algebra. Let $\mathfrak{A}_{\sa}$ denote the space of self-adjoint elements and write $a*b := \frac12(ab+ba)$ for the \Define{special Jordan product}. Then $(\mathfrak{A}_{\sa},*,1,\leq)$ where $\leq$ is the standard order of a C$^*$-algebra is a JB-algebra. Its norm is the regular C$^*$-norm.
\end{example}

\begin{example}
    A finite-dimensional JB-algebra $A$ is a \Define{Euclidean Jordan algebra} (EJA): a Jordan algebra equipped with an inner product $\inn{\cdot,\cdot}\colon A\times A\rightarrow \R$ such that $\inn{a*b,c} = \inn{b,a*c}$ for all $a,b,c\in A$ (and vice versa any Euclidean Jordan algebra is a JB-algebra).
    The Euclidean Jordan algebras have been fully classified~\cite{jordan1993algebraic}: they are direct sums of simple EJAs, and these are either matrix algebras over the real, complex or quaternionic fields, a type of algebra known as a \emph{spin factor}, or the so-called \emph{exceptional} algebra of $3\times 3$ Hermitian matrices over the octonions. Except for this last one, each of these algebras can be embedded into a C$^*$-algebra.
\end{example}




We are interested not in JB-algebras, but in \emph{JBW-algebras}, which are a class of JB-algebras that have more structure. They relate to JB-algebras in an analogous manner to how von Neumann algebras (i.e.~W$^*$-algebras) relate to C$^*$-algebras, hence the `W' in `JBW'.

\begin{definition}\label{def:order-separating}
    Let $A$ be an order unit space (such as a JB-algebra). A \Define{state} of $A$ is a positive unital linear map $\omega\colon A\rightarrow \R$. We say a state (or more generally any positive linear map) is \Define{normal} when it preserves suprema of directed sets: $\omega(\bigvee S) = \bigvee_{s\in S} \omega(s)$ for any directed~$S$. We say $A$ has a \Define{separating set} of normal states when for any two~$a,b\in A,a \neq b$ we can find a normal state $\omega$ such that $\omega(a)\neq \omega(b)$.
\end{definition}

\begin{definition}\label{def:JBW-algebra}
  A JB-algebra $A$ is a \Define{JBW-algebra} when it is directed complete and has a separating set of normal states. We denote by \textbf{JBW}$_{\text{pc}}$ the category of JBW-algebras with positive linear contractions, and by \textbf{JBW}$_{\text{npc}}$ for the wide subcategory of normal positive linear contractions.
\end{definition}


\begin{example}
    Let $\mathfrak{A}$ be a \Define{von Neumann algebra}, i.e.~a C$^*$-algebra that is directed complete and has a separating set of normal states~\cite{kadison1956operator}. Then its space of self-adjoint elements $\mathfrak{A}_{\sa}$ equipped with the special Jordan product is a JBW-algebra.
\end{example}

JBW-algebras are very close to the more familiar von Neumann algebras. Indeed, a large class of JBW-algebras comes from von Neumann algebras:
\begin{definition}
    A JBW-algebra $A$ is a \Define{JW-algebra} when it is
    Jordan-isomorphic to an ultraweakly
    closed subset of the
    self-adjoint elements of a von Neumann algebra.
\end{definition}

The counterpart to such `well-behaved' algebras are the \emph{exceptional} Jordan algebras.

\begin{definition}
    Let $A$ be a JB-algebra. We call $A$ \Define{purely exceptional}
  when any Jordan homomorphism $\phi\colon A\rightarrow \mathfrak{A}_{\sa}$ into a C$^*$-algebra $\mathfrak{A}$ is necessarily zero.
\end{definition}

\begin{theorem}[{\cite[Theorem 7.2.7]{hanche1984jordan}}]\label{thm:JBW-decomposition}
    Let $A$ be a JBW-algebra. Then there is a unique decomposition $A=A_{\text{sp}} \oplus A_{\text{ex}}$ where~$A_{\text{sp}}$ is a JW-algebra and~$A_{\text{ex}}$ is a purely exceptional JBW-algebra.
\end{theorem}

Interestingly, purely exceptional JBW-algebras only come in one type. To state this result, we need some more definitions.

\begin{definition}
    Let $X$ be a Stonean space (i.e.~an extremally disconnected compact Hausdorff space). We call $X$ \Define{hyperstonean} when $C(X,\R):=\{f:X\rightarrow \R~\text{continuous}\}$ is separated by normal states.
\end{definition}

Note that a compact Hausdorff space $X$ is hyperstonean if and only if $C(X,\R)$ is an associative JBW-algebra (or equivalently when $C(X,\CC)$ is a commutative von Neumann algebra). 

\begin{example}[{\cite{shultz1979normed}}]
    Let $X$ be a hyperstonean space and let $E = M_3(\mathbb{O})_{\sa}$ denote the exceptional Albert algebra of $3\times 3$ self-adjoint matrices of octonions $\mathbb{O}$ equipped with the standard Jordan product. Denote by $C(X,E)$ the set of continuous functions $f\colon X\rightarrow E$. Then $C(X,E)$ is a purely exceptional JBW-algebra with
    the Jordan product given pointwise by~$(f*g)(x) = f(x)*g(x)$.
\end{example}


\begin{theorem}[{\cite{shultz1979normed}}]\label{thm:purely-exceptional-char}
    Let $A$ be a purely exceptional JBW-algebra. Then there
    exists a hyperstonean space $X$, such that~$A\cong C(X,M_3(\mathbb{O}))$.
\end{theorem}

Combining Theorems~\ref{thm:JBW-decomposition}
and~\ref{thm:purely-exceptional-char} we see that any JBW-algebra
splits up into a part that embeds into a von Neumann algebra and a
part that is characterised by a hyperstonean space.

\subsection{Sequential products}\label{sec:seqprod}

Our reconstruction relies heavily on the categorical structures outlined in Section~\ref{sec:effectus}.
However, we will also need some assumptions that are of a more operational nature.
In particular, we consider the operation of `measuring' a predicate. This will take the form of a self-map $\asrt_p\colon A\rightarrow A$ for each predicate $p\in\Pred(A)$ that `asserts' that $p$ is true.
For an effect $p\in B(\mathcal{H})$ on a Hilbert space $\mathcal{H}$ this map is of the form $\asrt_p(q) = \sqrt{p}q\sqrt{p}$.
When given a set of assert maps for each predicate we denote $p\mult q:=q\circ \asrt_p$ for the \emph{sequential product} that can be interpreted as `observe $p$ and then observe $q$'~\cite{gudder2001sequential}.
We will require \& to satisfy a number of assumptions that will make $\Pred(A)$ into a \emph{sequential effect algebra}~\cite{gudder2002sequential}. Before we give the definition, let us motivate some of these conditions.

The sequential product $p\mult q$ of two effects $p$ and $q$ represents the sequential measurement of first $p$ and then $q$.
An important difference between classical and quantum systems is that in a classical system we can measure without disturbance, and hence the order of measurement is not important: $p\mult q = q\mult p$ for all predicates $p$ and $q$.
In a quantum system this is generally not the case, and the order of measurement is important (indeed, this is essentially Heisenberg uncertainty.)
However, what is interesting in quantum theory is that some
measurements are \emph{compatible}, meaning that the order of
measurement for those measurements is not important.
We will use the symbol $p\commu q$ to denote that~$p \mult q = q \mult p$.

\begin{definition}
  \label{defn:sea}
    A \Define{sequential effect algebra}
    (SEA)~\cite{gudder2002sequential}~$E$
    is an
     effect algebra with an additional (total)
    binary operation~$\mult$,
        called the \Define{sequential product},
        satisfying the axioms listed below,
        where $a,b,c\in E$.
        Elements~$a$ and~$b$ are said to \Define{commute},
            written~$a \commu b$,
            whenever~$a \mult b = b \mult a$.
    \begin{enumerate}[a)]
        \item
            $a\mult (b\ovee c) = a\mult b \ovee a \mult c$
            whenever~$b\perp c$.
        \item
            $1\mult a = a$.
        \item
            $a\mult b = 0 \implies b\mult a =0$.
        \item
            If $a\commu b$, then $a\commu b^\perp$ and $a\mult
                (b\mult c) = (a\mult b)\mult c$ for all $c$.
        \item
                If $c\commu a$ and $c\commu b$ then also $c\commu
        a\mult b$ and if furthermore $a\perp  b$, then $c\commu a\ovee
        b$.
    \end{enumerate}
    A SEA~$E$ is called \Define{normal}
        when~$E$ is directed complete, and
    \begin{enumerate}[f)]
        \item
            Given directed~$S\subseteq E$ we have
            $a\mult \bigvee S = \bigvee_{s\in S} a\mult s$,
            and $a\commu \bigvee S$ when $a\commu s$ for all~$s\in S$.
    \end{enumerate}
\end{definition}

Normal SEAs were studied in~\cite{second}, where they were shown to have many desirable properties. 
Let us note some of these properties here for later reference. Note that we call an effect $p$ idempotent
when $p\mult p = p$.
\begin{lemma}[{\cite{second}}]\label{lem:normal-SEA-properties}
    Let $E$ be a normal SEA and let $a,b\in E$ be arbitrary. Then the following are true
    \begin{enumerate}[a)]
        \item There is a smallest idempotent effect above $a$, which we denote by $\ceil{a}$.
        \item There is a largest idempotent effect beneath $a$, which we denote by $\floor{a}$.
        \item If $b\mult a = a$, then $b\geq \ceil{a}$.
    \end{enumerate}
\end{lemma}
\begin{remark}
In~\cite{second} it is also shown that normal SEAs satisfy a spectral theorem, 
and that in particular every effect can be written as a supremum and norm-limit 
of a sequence of \emph{simple} effects, effects that are finite linear combinations 
of idempotent effects. This implies that the sharp effects span a norm-dense set
of effects. These properties will sometimes implicitly be used in our reconstruction,
in particular in proving Lemma~\ref{lem:state-order-lemma}.
\end{remark}

The unit interval of a JBW-algebra is an example of a normal SEA. This is defined in terms of the \Define{quadratic product}. Let $A$ be a JBW-algebra and $a,b\in A$ arbitrary. Then we define $Q_a\colon A\rightarrow A$ as the map $Q_a b = 2a*(a*b)-a^2*b$. While this might look arbitrary, when $A$ is a JW-algebra this boils down to $Q_ab = aba$ using the product in the underlying von Neumann algebra. It is then perhaps not too surprising that the operation $a\mult b := Q_{\sqrt{a}} b$ on the unit interval of a JBW-algebra defines a normal sequential product~\cite{wetering2019commutativity}.

What is perhaps more surprising is that there is also a converse to this.
As it will inform the structure of our reconstruction proof, let us now recall two properties introduced in~\cite{wetering2018sequential} that force a convex normal SEA to have a Jordan algebra structure. 
Analogously to the definition for order unit spaces, we call a map $\omega\colon E\rightarrow [0,1]$ for a convex SEA $E$ a state when~$\omega$ is linear (i.e.~additive and preserves the scalar
    multiplication) and $\omega(1)=1$.

\begin{definition}\label{def:SEA-compressible}
    We say the sequential product of a convex SEA $E$ is
    \Define{compressible} when for all idempotent effects $p\in E$
    the following implication holds for all states~$\omega\colon E\rightarrow
    [0,1]$: if~$\omega(p) = 1$, then~$\omega(p\mult a) = \omega(a)$
    for all $a\in E$.
\end{definition}

What this property says is that if an effect $p$ already holds with certainty on a state~$\omega$, then measuring~$p$ does not affect the probabilities of other effects holding in the state~$\omega$.%
\footnote{The work of Alfsen and Shultz uses the notion of a \emph{compression}, which is a special type of an idempotent map~\cite{alfsen2012geometry}.
The multiplication maps $L_p(a) = p\mult a$ of convex SEAs always satisfy three of the four conditions of being a compression. The fourth condition, namely the implication $\omega\circ L_{p^\perp} = 0 \implies \omega\circ L_p = \omega$, is satisfied iff the SEA is compressible, hence the name.
}
The second property is a weaker version of the \emph{fundamental identity of quadratic Jordan algebras}~\cite{racine1973arithmetics}.
\begin{definition}\label{def:SEA-quadratic}
    We say the sequential product of a SEA $E$ is \Define{quadratic} when for any two idempotents $p,q\in E$ we have $q\mult(p\mult q) = (q\mult p)^2$.
\end{definition}

\begin{theorem}[{\cite[Theorem~4]{wetering2018sequential}}]\label{thm:normalSEAisJB}
    Let $E$ be a convex normal SEA and suppose its sequential product is compressible and quadratic. Then $E$ is order-isomorphic to the unit interval of a directed-complete JB-algebra.%
    \footnote{As far as the authors are aware, there is no known example of a convex sequential effect algebra that is not compressible, nor one that is not quadratic. Hence, it might be that these properties hold for all convex SEAs and thus that the conditions in this theorem can be simplified.}
\end{theorem}

\section{Pure maps and \texorpdfstring{$\diamond$}{diamond}-adjointness}\label{sec:diamond}

With the necessary concepts introduced we will now take a closer look at filters and comprehensions, which form the backbone of our reconstruction.
These will allow us to define the concepts of pure maps and $\diamond$-adjointness in an effectus.

To get a better handle on comprehensions, we introduce the \emph{image} of
    a map.  If we have a filter, then we can extract the predicate being filtered by applying~$\truth \after (\ )$. Analogously, to get the predicate of a comprehension we ask for its image.

\begin{definition}
    Let $f\colon A\rightarrow B$ be a morphism in an effectus. The
    \Define{image} of $f$, when it exists, is the smallest predicate
    $p\colon B\rightarrow I$ such that $p\circ f = \truth\circ f$, i.e.\
    if $q\colon B\rightarrow I$ is also such that $q\circ f = \truth\circ
    f$, then $p\leq q$. We denote the image of $f$ by $\im{f}$. We
    say an effectus  \Define{has images} when all the maps have
    an image.
\end{definition}

\begin{lemma}\label{lem:imageofcomposedmaps}
  Let $f$ and $g$ be composable maps and suppose $\im{f\circ g}$ and $\im{f}$ exist. Then $\im{(f\circ g)} \leq \im{f}$. Furthermore, if $g$ is an isomorphism, then $\im{(f\circ g)} = \im{f}$.
\end{lemma}
\begin{proof}
  We of course have $1\circ (f\circ g) = (1\circ f)\circ g = (\im{f}\circ f)\circ g = \im{f}\circ (f\circ g)$, and hence $\im{f}\leq \im{f\circ g}$.

  If $g$ is an isomorphism, then we furthermore have $\im{f} = \im{(f\circ g)\circ g^{-1}} \leq \im{f\circ g} \leq \im{f}$, and hence $\im{f} = \im{f\circ g}$.
\end{proof}

\begin{definition}\label{def:effectus-sharp}
    Let $p\colon A\rightarrow I$ be a predicate. We call $p$ \Define{sharp} when there is some morphism $f\colon B\rightarrow A$ such that $\im{f} = p$. We define SPred$(A)$ to be the poset of sharp predicates of $A$.
\end{definition}

When we have images and comprehensions, we can find for each predicate the largest sharp predicate beneath it.

\begin{definition}\label{def:floorceiling}
    The \Define{floor} of $p$ is defined as $\floor{p}:= \im{\pi}$,
        where~$\pi$ is any comprehension for~$p$.%
        \footnote{%
                This is well-defined
                    as for any two comprehensions~$\pi,\pi'$ for the same predicate,
                        there exists an isomorphism~$\alpha$
                        with~$\pi = \pi' \circ \alpha$
                        so that~$\im \pi \leq \im \pi' \leq \im \pi $
                        by Lemma~\ref{lem:imageofcomposedmaps}.
        }
    The \Define{ceiling} is defined as the De Morgan dual: $\ceil{p} :=
    \floor{p^\perp}^\perp$.
\end{definition}

Note that in this general setting it is not necessarily the case that $\ceil{p}$ is sharp.
This needs to be imposed additionally (see next section).

\begin{proposition}\label{prop:floorceiling}
  In an effectus with images and compressions, the following are true for any predicates $q\leq p$ and composable map $f$.
  \begin{multicols}{2}
  \begin{enumerate}[a)]
    \item $\floor{p}\leq p \leq \ceil{p}$.
    \item $\floor{\floor{p}}=\floor{p}$.
    \item $\floor{q}\leq \floor{p}$ and $\ceil{q}\leq \ceil{p}$.
    \item $\ceil{p\circ f} = \ceil{\ceil{p}\circ f}$.
    \item $\ceil{p}\circ f = 0 \iff p\circ f = 0$.
    \item $p$ is sharp if and only if $\floor{p}=p$.
  \end{enumerate}
  \end{multicols}
\end{proposition}
\begin{proof}
  Let $p, q\colon  A\rightarrow I$ be predicates, and let $\pi_p\colon A_p \rightarrow A$ be a comprehension for $p$, and $\pi_q\colon A_q \rightarrow A$ a compression for $q$.
  \begin{enumerate}[a)]
    \item Of course $1\circ \pi_p = p\circ \pi_p$, and hence $\floor{p} := \im{\pi_p} \leq p$. Hence also $\floor{p^\perp} \leq p^\perp$, and thus $\ceil{p} := \floor{p^\perp}^\perp \geq (p^\perp)^\perp = p$.
    
    \item We will show that $\pi_p$ is a compression for $\floor{p}$, and hence $\pi_{\floor{p}} = \pi_p \circ\Theta$ for some isomorphism $\Theta$. The result then follows using Lemma~\ref{lem:imageofcomposedmaps}, because $\floor{\floor{p}} := \im{\pi_{\floor{p}}} = \im{\pi_p\circ \Theta} = \im{\pi_p} = \floor{p}$.

    Note first that $\floor{p}\circ \pi_p = \im{\pi_p}\circ \pi_p = 1\circ \pi_p$. Now let $f\colon B\rightarrow A$ be some map with $\floor{p}\circ f = 1\circ f$. As $\floor{p}\leq p$, we then also have $p\circ f = 1\circ f$, and hence by the universal property of $\pi_p$ there is a unique $\cl{f}$ with $f=\pi_p \circ \cl{f}$. Hence, $\pi_p$ is also a compression for $\floor{p}$.

    \item We have $1\circ \pi_q = q\circ \pi_q \leq p\circ \pi_q \leq 1\circ \pi_q$, and hence $p\circ \pi_q = 1\circ \pi_q$. Hence $\pi_q = \pi_p \circ \cl{\pi_q}$ for a unique $\cl{\pi_q}$. With Lemma~\ref{lem:imageofcomposedmaps} we calculate $\floor{q} := \im{\pi_q} = \im{\pi_p\circ\cl{\pi_q}} \leq \im{\pi_p} = \floor{p}$. To show $\ceil{q}\leq \ceil{p}$, we note that as $q\leq p$, we have $p^\perp \leq q^\perp$ and hence $\floor{p^\perp}\leq \floor{q^\perp}$. Then: $\ceil{q}:=\floor{q^\perp}^\perp \leq \floor{p^\perp}^\perp = \ceil{p}$.

    \item First note that since $\ceil{p}\circ f \geq p\circ f$, we have by point c): $\ceil{\ceil{p}\circ f} \geq \ceil{p\circ f}$. It remains to show the inequality in the other direction.

    Because $p\circ (f\circ \pi_{(p\circ f)^\perp}) = 0$, there must be an $h$ with $f\circ \pi_{(p\circ f)^\perp} = \pi_{p^\perp}\circ h$. By point b) there must be some isomorphism $\Theta$ such that $\pi_{p^\perp} = \pi_{\floor{p^\perp}} \circ \Theta = \pi_{\ceil{p}^\perp} \circ \Theta$. We then calculate:
    $$\ceil{p}\circ f\circ \pi_{(p\circ f)^\perp} \ =\  \ceil{p}\circ \pi_{p^\perp}\circ h \ =\  \ceil{p} \circ \pi_{\ceil{p}^\perp} \circ \Theta\circ h \ =\  0.$$
    Hence $\ceil{p}\circ f \leq \im{\pi_{(p\circ f)^\perp}}^\perp = \floor{(p\circ f)^\perp}^\perp = \ceil{p\circ f}$. Using points c) and b): $\ceil{\ceil{p}\circ f}\leq \ceil{\ceil{p\circ f}} = \ceil{p\circ f}$. 

    \item Of course if $\ceil{p}\circ f = 0$, then $p\circ f \leq \ceil{p}\circ f = 0$. For the other direction, we remark that $\ceil{0} = \floor{1}^\perp = 1^\perp = 0$, so that by the previous point: $0 = \ceil{0} = \ceil{p\circ f} = \ceil{\ceil{p}\circ f}$, and hence $\ceil{p}\circ f \leq \ceil{\ceil{p}\circ f} = 0$.

    \item If $\floor{p} = p$, then $p=\im{\pi_p}$, and hence $p$ is sharp. Now suppose $p$ is sharp, and hence is the image of some map $f$: $p=\im{f}$. Then by the universal property of $\pi_p$, there is some $\cl{f}$ such that $f=\pi_p \circ \cl{f}$. We then calculate using Lemma~\ref{lem:imageofcomposedmaps} $p=\im{f} = \im{\pi_p\circ \cl{f}} \leq \im{\pi_p} = \floor{p}$. As $\floor{p}\leq p$ by point a), we are done. \qedhere
  \end{enumerate}
\end{proof}


\subsection{\texorpdfstring{$\diamond$}{Diamond}-effectuses}\label{sec:diamond-effectus}

In the setting where we have filters, comprehensions and images, the floor of a predicate is already well-behaved, but the ceiling is not necessarily. In particular, it is not necessarily the case that $\ceil{p}$ is sharp. To ensure this we must require that $p$ is sharp iff $p^\perp$ is sharp. This gives us the following definition.

\begin{definition}[{\cite[Section 3.5]{basthesis}}]\label{def:diamond-effect-theory}
  An effectus is a \Define{$\diamond$-effectus} (pronounced `diamond-effectus') when it has images, filters, comprehensions and if a predicate $p$ is sharp iff $p^\perp$ is sharp.
\end{definition}

The reason we call it a $\diamond$-effectus, is because of the `possibilistic' structure that is present in such an effectus.
\begin{definition}
  Let $A$ and $B$ be objects in a $\diamond$-effectus. 
  For any $f\colon A\rightarrow B$ we define 
  \[f^\diamond\colon \text{SPred}(B)\rightarrow \text{SPred}(A) \ \  \text{and} \ \  f_\diamond\colon\text{SPred}(A)\rightarrow \text{SPred}(B)\]
  by $f^\diamond(p) := \ceil{p \circ f}$ and $f_\diamond(p) := \im(f\circ\pi_p)$,
         where $\pi_p$ is any compression for $p$.
\end{definition}

It will be useful to introduce a third such map:
$f^{\ssquare}(p) := (f^\diamond(p^\perp))^\perp$.
These maps forget the exact probabilities involved in~$f$
    (indeed, for instance~$(\frac{1}{2}f)^\diamond = f^\diamond$),
    and only remember what is possible.

\begin{proposition}\label{prop:galois-properties}
  Let $f\colon A\rightarrow B$ and $g\colon B\rightarrow C$ be maps in a $\diamond$-effectus and let $p\in \text{SEff}(B)$ and $q\in \text{SEff}(A)$. Then the following are true.
  \begin{multicols}{2}
  \begin{enumerate}[a)]
    \item $f^\diamond$ and $f^{\ssquare}$ are monotone.
    \item $f^\diamond(p)\leq q^\perp \iff f_\diamond(q)\leq p^\perp$.
    \item $f_\diamond$ and $f^{\ssquare}$ form a Galois connection.
    \item $f_\diamond$ is monotone.
    \item $f_\diamond\circ f^{\ssquare}\circ f_\diamond = f_\diamond$.
    \item $(\id)^\diamond = (\id)_\diamond = (\id)^{\ssquare} = \id$.
    \item $(f\circ g)^\diamond = g^\diamond \circ f^\diamond$, $(f\circ g)^{\ssquare} = g^{\ssquare} \circ f^{\ssquare}$.
    \item $(f\circ g)_\diamond = f_\diamond \circ g_\diamond$.
  \end{enumerate}
  \end{multicols}
\end{proposition}
\begin{proof}~
  \begin{enumerate}[a)]
    \item Suppose $p\leq q$. Then $p\circ f\leq q\circ f$ and hence $f^\diamond(p) = \ceil{p\circ f} \leq \ceil{q\circ f} = f^\diamond(q)$. Also $q^\perp \leq p^\perp$ and hence $\ceil{q^\perp \circ f} \leq \ceil{p^\perp\circ f}$. Taking complements again gives $f^{\ssquare}(p)\leq f^{\ssquare}(q)$.
    
    \item Suppose $f^\diamond(p)\leq q^\perp$. Then $p\circ f \leq \ceil{p\circ f} = f^\diamond(p) \leq q^\perp = \im{\pi_q}^\perp$ and hence $p\circ f \circ \pi_q = 0$ so that $p\leq \im{f\circ \pi_q}^\perp$. But then $f_\diamond(q) = \im{f\circ \pi_q} \leq p^\perp$.
    
    \item Suppose $f_\diamond(q)\leq p$. We need to show $q\leq f^{\ssquare}(p)$. The previous point gives $f_\diamond(q)\leq p$ iff $f^\diamond(q^\perp) \leq p^\perp$ and hence $p\leq f^\diamond(q^\perp)^\perp =: f^{\ssquare}(q)$.
    
    \item As $f_\diamond(q)\leq f_\diamond(q)$ the previous point gives $q\leq f^{\ssquare}(f_\diamond(q))$. Now suppose $p\leq q$. Then $p\leq q\leq f^{\ssquare}(f_\diamond(q))$ so that again by the previous point $f_\diamond(p)\leq f_\diamond(q)$.

    \item We have $f_\diamond(f^{\ssquare}(p))\leq p$. Leting $p:= f_\diamond(q)$ we get $f_\diamond f^{\ssquare} f_\diamond (q) \leq f_\diamond(q)$. We also have $q\leq f^{\ssquare} f_\diamond (q)$ and hence applying the monotone $f_\diamond$ to both sides gives the other inequality.

    \item For sharp $p$ we have $\ceil{p} = p = \im{\pi_p}$. The statements then follow easily.
    \item By Proposition~\ref{prop:floorceiling}: $(f\circ g)^\diamond(p) = \ceil{p\circ f\circ g} = \ceil{\ceil{p\circ f}\circ g} = g^\diamond(f^\diamond(p))$. Furthermore $(f\circ g)^{\ssquare}(p) = (f\circ g)^\diamond(p^\perp)^\perp = g^\diamond(f^\diamond(p^\perp))^\perp = g^\diamond(f^{\ssquare}(p)^\perp)^\perp = g^{\ssquare}(f^{\ssquare}(p))$.
    \item $(f\circ g)_\diamond$ is left Galois adjoint to $(f\circ g)^{\ssquare}$. As Galois adjoints are unique it suffices to show that $f_\diamond\circ g_\diamond$ is also left Galois adjoint to $(f\circ g)^{\ssquare}$. We calculate:
    \[f_\diamond(g_\diamond(p))\ \leq\  q \ \iff\  g_\diamond(p) \ \leq\  f^{\ssquare}(q) \ \iff\  p \ \leq\  g^{\ssquare}(f^{\ssquare}(q))\  =\  (f\circ g)^{\ssquare}(q). \qedhere\]
  \end{enumerate}
\end{proof}

The importance of the~$\diamond$-structure
    is that it allows us to have a notion of adjointness
    between arbitrary processes, somewhat like a `dagger':

\begin{definition}
    We say maps $f\colon A\rightarrow B$ and $g\colon B\rightarrow A$ in a $\diamond$-effectus are \Define{$\diamond$-adjoint} when $f^\diamond = g_\diamond$. An endomap $f\colon A\rightarrow A$ is \Define{$\diamond$-self-adjoint} when $f^\diamond = f_\diamond$.
\end{definition}

\begin{lemma}
    $\diamond$-adjointness is a symmetric relation: $f^\diamond = g_\diamond$ iff $g^\diamond = f_\diamond$.
\end{lemma}
\begin{proof}
    Suppose $f^\diamond = g_\diamond$.
    Using Proposition~\ref{prop:galois-properties}.b) twice, we calculate:
    \[f_\diamond(p)\ \leq\  q^\perp \ \iff\  f^\diamond(q)\ \leq\  p^\perp\ \iff\ g_\diamond(q)\ \leq\  p^\perp\ \iff\ g^\diamond(p) \ \leq\  q^\perp.\]
    Then take $q:= f_\diamond(p)^\perp$ and $q:=g^\diamond(p)^\perp$ to get inequalities in both directions that proves
    $f_\diamond(p) = g^\diamond(p)$.
\end{proof}

\begin{remark}
In general, a map does not have a unique~$\diamond$-adjoint.
A trivial reason is that~$f^\diamond = (\frac{1}{2}f)^\diamond$.
Interestingly, even if~$f^\diamond = g^\diamond$
    and~$\truth\after f = \truth \after g$, then it does not have to be the case that~$f = g$.
If the previous is true (with fixed~$f$ for arbitrary~$g$),
    then we say~$f$ is \Define{rigid}.
Rigidity for maps between von Neumann algebras is studied in~\cite[\S102]{bramthesis}.
\end{remark}

\begin{example}
For a von Neumann algebra $\mathfrak{A}$ and an element $a\in\mathfrak{A}$, the conjugation endomaps $b\mapsto a^*ba$ and $b\mapsto aba^*$ on $\mathfrak{A}$ are $\diamond$-adjoint and so are the standard comprehension and standard filter of a given projection (see~Ex.~\ref{ex:quotient-comprehension})~\cite[\S 101]{bramthesis}.
\end{example}

We study some additional structure of $\diamond$-effectuses that is not relevant to the reconstruction but might be of interest in its own right in Section~\ref{sec:diamond-additional}.

\subsection{Compatible filters and comprehensions}

So far we have not required any type of coherence between filters and comprehensions. For the next results we however do need to know a bit more about their interplay.
 
\begin{definition}\label{def:compatible-quotients-comprehensions}
  Let $\catC$ be an effectus with filters and comprehensions. We say the filters and comprehensions are \Define{compatible} when for every comprehension $\pi_p$ of a sharp predicate $p$ there exists a filter $\xi^p$ of $p$ such that $\xi^p\circ \pi_p = \id$.
\end{definition}
In the setting we care about, there will be a dagger structure on some of the maps in an effectus which in particular will entail that $\pi_p^\dagger$ is a filter with $\pi_p^\dagger \circ \pi_p = \id$, so that they are indeed compatible.

In an effectus with compatible filters and comprehensions we can for every sharp predicate define a special type of idempotent map that acts as a measurement for this predicate.

\begin{definition}\label{def:sharp-assert}
    Let $\catC$ be an effectus with compatible filters and comprehensions, and let $p\colon A\rightarrow I$ be a sharp predicate.
    We define the \Define{assert map} $\asrt_p\colon A\rightarrow A$ for $p$ to be $\asrt_p := \pi_p\circ \xi^p$ where $\pi_p$ and $\xi^p$ are a compatible pair of a comprehension and filter for $p$.
\end{definition}

Note first that assert maps are uniquely defined, since if $\pi_p'$ and $(\xi^p)'$ are another compatible pair, then $\pi_p' = \pi_p\circ \Theta_1$ and $(\xi^p)' = \Theta_2\circ \xi^p$ for some isomorphisms $\Theta_1$ and $\Theta_2$ and furthermore $\id = (\xi^p)'\circ \pi_p' = \Theta_2\circ \xi^p\circ\pi_p \circ \Theta_1 = \Theta_2 \circ \Theta_1$ so that $\Theta_2 = \Theta_1^{-1}$ and hence $\pi_p'\circ (\xi^p)' = \pi_p\circ \Theta_1\circ \Theta_1^{-1} \circ \xi^p = \pi_p\circ \xi^p = \asrt_p$.
Additionally, since the filter and comprehension are compatible we get $\asrt_p\circ \asrt_p = \asrt_p$.

\begin{example}
    Let $B(\mathcal{H})$ be an object in the opposite category of C$^*$-algebras. A sharp predicate then corresponds to a projector $P:\mathcal{H}\rightarrow \mathcal{H}$. The assert map for $P$ is then given by $\asrt_P(A) = PAP$.
\end{example}

Note that we have $1\circ \asrt_p = \im{\asrt_p} = p$.

\begin{lemma}\label{lem:assert-image}
    Let $\catC$ be an effectus with images and compatible filters and comprehensions.
    Let $p\in\Pred(A)$ be a sharp predicate and let $f\colon B\rightarrow A$ and $g\colon A\rightarrow B$ be morphisms in the effectus. The following are true:
    \begin{enumerate}[a)]
        \item $\im{f}\leq p \iff \asrt_p\circ f = f$.
        \item $1\circ g \leq p \iff g\circ \asrt_p = g$.
    \end{enumerate}
\end{lemma}
\begin{proof}
    For the first point:
    if $\asrt_p\circ f = f$, then $\im{f} = {\im{\asrt_p\circ f}} \leq \im{\asrt_p} = p$. Conversely, if $\im{f}\leq p$, then $p\circ f = 1\circ f$ so that by the universal property of $\pi_p$ we have $f = \pi_p\circ \cl{f}$ for some $\cl{f}$. Now $\cl{f} = \id\circ \cl{f} = \xi^p\circ \pi_p \circ \cl{f} = \xi^p \circ f$ so that $f = \pi_p \circ \cl{f} = \pi_p\circ \xi^p\circ f = \asrt_p\circ f$.

    For the second point:
    suppose $g\circ \asrt_p = g$. Then $1\circ g = (1\circ g)\circ \asrt_p \leq 1\circ \asrt_p = p$. Conversely, if $1\circ g \leq p$, then by the universal property of $\xi^p$ we have $g = \cl{g}\circ \xi^p$ for some $\cl{g}$. Now $\cl{g} = \cl{g}\circ \id = \cl{g} \circ \xi^p \circ \pi_p = g \circ \pi_p$ so that $g = \cl{g}\circ \xi^p = g\circ \pi_p\circ \xi^p = g\circ \asrt_p$.
\end{proof}

\subsection{Pure maps}\label{sec:pure-maps}
In pure finite-dimensional quantum theory there is
    a well-established notion of \emph{pure map}:
        Kraus rank-1, i.e.~a map of the form~$T \mapsto A^* T A$ for some operator $A$.
        These correspond to the evolution of a system
            that does not include measurement (but may include
            set-up or loss of knowledge).
For infinite-dimensional systems, there are many proposals,
    but we argue for the following (cf.~\cite[\S168]{basthesis}).

\begin{definition}[{\cite[Definition~201II]{basthesis}}]\label{def:pure}
    Let  $f\colon A\rightarrow B$ be a map in an effectus. We say it is \Define{pure} when $f=\pi\circ\xi$ for some filter $\xi$ and comprehension $\pi$.
\end{definition}

\begin{example}
    In the category of von Neumann algebras with normal completely positive
        contractive linear maps in the opposite direction,
        a map~$f\colon B(\mathcal{H}) \to B(\mathcal{K})$
            is pure iff it is Kraus rank-1, as desired.
    In fact, every map factors (in a Stinespring-like fashion)
        as a pure map after a
        normal~$*$-homomorphism~\cite{westerbaan2016paschke}
        and a map is pure iff this $*$-homomorphism is surjective.
\end{example}

Note that the ordering of the filter and comprehension $\pi\circ\xi$ is important. This raises a question whether pure maps are closed under composition. Some compositions always result in a pure map again.

\begin{lemma}
    Let $\catC$ be an effectus with filters. Then the following are true.
    \begin{itemize}
        \item A composition of filters is again a filter.
        \item If $\catC$ also has images and compatible comprehensions, then a composition of comprehensions is again a comprehension.
    \end{itemize}
\end{lemma}
\begin{proof}
    Let $\xi^p$ and $\xi^q$ be filters for $p$ respectively $q$. We claim that $\xi^p\circ\xi^q$ is a filter for $p\circ\xi_q$. To this end we let $\xi$ be a filter for $p\circ\xi^q$. Then there is a unique $g$ such that $\xi^p\circ\xi^q = g\circ \xi$, which we need to show is an isomorphism. As $1\circ\xi = p\circ\xi^q \leq 1\circ\xi^q=q$ we have $\xi = h_1\circ\xi^q$ for a unique $h_1$. Because $1\circ h_1\circ \xi^q = 1\circ\xi = p\circ \xi^q$ we have $1\circ h_1 = p$ because $\xi^q$ is epic and hence $h_1 = h_2\circ \xi^p$. Then
    \[g\circ h_2\circ \xi^p\circ \xi^q \ =\  g\circ \xi \ =\  \xi^p\circ \xi^q\quad \text{and} \quad h_2\circ g \circ \xi \ =\  h_2 \circ \xi^p\circ \xi^q \ =\  h_1\circ \xi^q \ =\  \xi\]
  so that because $\xi^p\circ \xi^q$ and $\xi$ are epic we have $g\circ h_2 = \id$ and $h_2\circ g = \id$.

  Now suppose $\catC$ has images and compatible comprehensions and let $\pi_p$ and $\pi_q$ be comprehensions for sharp predicates $p$ respectively $q$. 
  We will show that $\pi_p\circ \pi_q$ is a comprehension for $\im{(\pi_p \circ \pi_q)}$. 
  To this end let $f$ be any map with $\im{\pi_p\circ \pi_q}\circ f = 1\circ f$. 
  As $\im{\pi_p\circ \pi_q} \leq \im{\pi_p} = p$ we also have $p\circ f = 1\circ f$ and hence $f = \pi_p\circ g_1$ for a unique $g_1$. 
  Let $\xi^p$ be a filter for $p$ such that $\xi^p\circ \pi_p = \id$. 
  Then $q\circ \xi^p\circ\pi_p\circ \pi_q = q\circ \pi_q = 1\circ\pi_q = 1\circ \pi_p \circ \pi_q$ and hence $q\circ \xi^p \geq \im{\pi_p\circ \pi_q}$ so that
  \[q\circ g_1 \ =\  q\circ \xi^p\circ \pi_p \circ g_1 \ =\  q\circ \xi^p \circ f \ \geq\  \im{\pi_p\circ \pi_q} \circ f \ =\  1\circ f \ =\  1\circ g_1.\]
  Hence there is a unique $g_2$ such that $g_1 = \pi_q\circ g_2$ which gives $f = \pi_p\circ g_1 = (\pi_p\circ \pi_q) \circ g_2$. Uniqueness of $g_2$ with the property that $f=(\pi_p\circ \pi_q)\circ g_2$ follows because $\pi_p\circ \pi_q$ is monic.
\end{proof}

Hence, the question whether a composition of pure maps is again pure is reduced to
the question whether a composition~$\xi\circ\pi$ `in the wrong order' can be written as $\pi'\circ\xi'$ for some different comprehension $\pi'$ and filter $\xi'$.
This is true in our main examples from quantum theory (indeed the composition of two Kraus rank-1 maps is again Kraus rank-1). For an effectus it is something that needs to be imposed additionally, i.e.~by demanding that the pure maps form a subcategory. 

In fact, inspired by quantum theory, for our reconstruction we will require that the pure maps form a \emph{dagger category} where for each pure map~$f\colon A\rightarrow B$ we have a pure map $f^\dagger\colon B\rightarrow A$ such that $(f^\dagger)^\dagger = f$ and $(f\circ g)^\dagger = g^\dagger\circ f^\dagger$.
Recall that in a dagger category we say $f$ is \Define{$\dagger$-adjoint} to $g$ when $f^\dagger = g$ and we say $f$ is an \Define{$\dagger$-isometry} when $f^\dagger\circ f = \id$. Finally, we say $f$ is~\Define{$\dagger$-positive} when $f=g\circ g^\dagger$ for some $g$.

\begin{remark}\label{rem:reconstruction-finite}
In Ref.~\cite{wetering2018reconstruction} a reconstruction of finite-dimensional quantum theory is given that can be stated using our language (up to some details) as follows:
Let $\catC$ be a state-separated $\diamond$-effectus with finite tomography where $\Pred(I)\cong [0,1]$ such that
\begin{itemize}
    \item the pure maps form a dagger category,
    \item a comprehension of a sharp predicate is $\dagger$-adjoint to a filter of the same predicate,
    \item and comprehensions are $\dagger$-isometries.
\end{itemize}
Then $\Pred(\catC)$ embeds into the category of \emph{Euclidean Jordan algebras} (and positive contractive linear maps).
If additionally $\catC$ is symmetric monoidal in a suitably compatible way, then $\Pred(\catC)$ embeds into the category of finite-dimensional C$^*$-algebras.
\end{remark}

In the remainder of the paper we will seek other conditions that suffice to get an analogous result, but that also includes infinite-dimensional spaces and doesn't require the assumption that $\Pred(I)\cong[0,1]$. Note that while our conditions (see Definition~\ref{def:sequential-effectus}) will be different, they will imply the points listed here.

\subsection{Further properties of \texorpdfstring{$\diamond$}{diamond}-effectuses}\label{sec:diamond-additional}

In this subsection we discuss some additional structure present in $\diamond$-effectuses that has no further bearing on our reconstruction, but might be of independent interest.

First, sharp predicates in a $\diamond$-effectus are also `sharp' in a more standard sense.
\begin{lemma}\label{lem:ortho-sharp}
    Let $p$ be a sharp predicate in a $\diamond$-effectus. Then $p$ is also \Define{ortho-sharp}, namely $p\wedge p^\perp = \falsity$.
\end{lemma}
\begin{proof}
    Let $q$ be any predicate on the same object as $p$. 
   Suppose $q\leq p$ and $q\leq p^\perp$. We need to show that $q=\falsity$.
   Note $\ceil{q}\leq \ceil{p} = \floor{p^\perp}^\perp = (p^\perp)^\perp = p$ (using Proposition~\ref{prop:floorceiling}.f)),
   and similarly $\ceil{q}\leq p^\perp$.
   Let $\pi_{\ceil{q}}$ and $\pi_p$ be comprehensions for $\ceil{q}$, respectively $p$. 
   Then by the universal property of $\pi_p$ there is a unique $f$ such that $\pi_{\ceil{q}} = \pi_p\after f$.
   We then calculate
   $$\truth\after\pi_{\ceil{q}} \ =\  \ceil{q}\after \pi_{\ceil{q}} \ =\  \ceil{q}\after \pi_p\after f \ \leq\  p^\perp\after \pi_p\after f \ =\  \falsity.$$
   Hence $\pi_{\ceil{q}}=0$ so that $q\leq \ceil{q} = \im{\pi_{\ceil{q}}} = \falsity$ as desired.
\end{proof}

\begin{lemma}
    Let $p,q\in \text{SPred}(A)$ be two sharp predicates in a $\diamond$-effectus. Then the following holds.
    \begin{itemize}
        \item The supremum $p\vee q$ in $\Pred(A)$ exists and is sharp.
        \item The sharp predicates SPred$(A)$ form an \Define{ortholattice}: a lattice with orthocomplement satisfying $p\wedge p^\perp = \falsity$ and $p\vee p^\perp = \truth$.
    \end{itemize}
\end{lemma}
\begin{proof}
    Let $p,q \in \text{SPred}(A)$. We claim that $p\wedge q = (\pi_p)_\diamond(\pi_p^{\ssquare}(q))$. First of all, as $q\leq q$, we have $\pi_p^{\ssquare}(q)\leq \pi_p^{\ssquare}(q)$ and thus $(\pi_p)_\diamond(\pi_p^{\ssquare}(q)) \leq q$. Second, $(\pi_p)_\diamond(\pi_p^{\ssquare}(q)) \leq (\pi_p)_\diamond(1) = \im{\pi_p} = p$, so it is indeed a lower bound. Now let $r$ be any sharp element with $r\leq p$ and $r\leq q$. 
    Then as in the previous point $\pi_r = \pi_p \circ h$ for some map $h$. Now using that $(\pi_p)_\diamond = (\pi_p)_\diamond\circ (\pi_p)^{\ssquare} \circ (\pi_p)_\diamond$; cf.~Proposition~\ref{prop:galois-properties}.e):
    \[(\pi_r)_\diamond \ =\  (\pi_p)_\diamond \circ h_\diamond\ =\ 
    (\pi_p)_\diamond\circ (\pi_p)^{\ssquare}\circ (\pi_p)_\diamond
    \circ h_\diamond \ =\  (\pi_p)_\diamond\circ (\pi_p)^{\ssquare}\circ
    (\pi_r)_\diamond\]
    and thus $r = (\pi_r)_\diamond(1) = (\pi_p)_\diamond\circ (\pi_p)^{\ssquare}(r) \leq (\pi_p)_\diamond\circ (\pi_p)^{\ssquare}(q)$.
    Now for a general $a\leq p,q$, we will also have $\ceil{a}\leq p,q$ and hence $a\leq \ceil{a} \leq p\wedge q$.

    Now to show SPred$(A)$ forms an ortholattice, first note that $(p^\perp\vee q^\perp)^\perp = p\wedge q$ so that is indeed a lattice.
    By Lemma~\ref{lem:ortho-sharp} we have $p\wedge p^\perp = \falsity$ and hence $p\vee p^\perp = (p^\perp \wedge p)^\perp = \falsity^\perp = \truth$.
\end{proof}

\begin{proposition}\label{prop:spred-is-oml}
    Let $A$ be an object in a $\diamond$-effectus. Then SPred$(A)$ is a sub-effect-algebra of $\Pred(A)$. Furthermore, SPred$(A)$ is an orthomodular lattice.
\end{proposition}
\begin{proof}
    SPred$(A)$ contains $\falsity$, $\truth$ and is closed under the complement so it remains to show that it is closed under sums.
    Let $p,q\in\text{SPred}(A)$ and suppose $p$ and $q$ are summable. We claim that $p\ovee q = p\vee q$ so that $p\ovee q$ is indeed sharp. We have $p\leq q^\perp$ and hence $p\wedge q \leq q^\perp\wedge q = \falsity$. That $p\ovee q = p\vee q$ then follows from the identity $a\ovee b = (a\wedge b)\ovee (a\vee b)$ that holds in effect algebras~\cite[Prop.~177]{basthesis}.
    
    Now any ortholattice that is also an effect algebra is an orthomodular lattice~\cite[Prop.~1.5.8]{dvurecenskij2013new}.
\end{proof}

\begin{definition}[{\cite{jacobs2009orthomodular}}]
    We define \textbf{OMLatGal} to be the category where the objects are orthomodular lattices and the morphisms are Galois connections $(f: A\rightleftarrows B: g $.
\end{definition}

\begin{proposition}
    Let $\catC$ be a $\diamond$-effectus. The assignment $A\mapsto \text{SPred}(A)$ and $f\mapsto (f_\diamond,f^{\ssquare})$ gives a functor $\catC\rightarrow \textbf{OMLatGal}$.
\end{proposition}
\begin{proof}
    SPred$(A)$ is an orthomodular lattice by Proposition~\ref{prop:spred-is-oml}. That $f_\diamond$ and $f^{\ssquare}$ form a Galois connection is Proposition~\ref{prop:galois-properties}.c), and that the assignment is functorial is given by point~g) of that same proposition.
\end{proof}

\section{Decomposing into sharp and convex systems}\label{sec:decomposition}

Our first step in reconstructing quantum theory is showing how we can retrieve convexity and real probabilities from the abstract framework of effectus theory. This relies mostly on showing how properties of the scalars of the effectus lift to the entirety of the category. We already saw examples of that in Section~\ref{sec:OUS}, where if the scalars were $\{0,1\}$ the predicate spaces were orthoalgebras, while if they were $[0,1]$, then the predicate spaces would be convex effect algebras. Here we will generalise these results.

\subsection{Decomposing an effectus}\label{sec:decomposing-effectus}

If the effect monoid of scalars of an effectus is reducible, then we can lift this to the level of the category.
This is because scalars in an effectus have an action on the predicate spaces through composition: given a scalar $s\colon I\rightarrow I$ and a predicate $a\colon A\rightarrow I$, we construct the scaled predicate $s\cdot a:= s\circ a$. 

\begin{proposition}\label{prop:effectus-product}
    Let $\catC$ be an effectus
        with a non-trivial idempotent scalar~$s$, i.e.~$s \notin \{0,1\}$ and~$s^2=s$.
    Then $\Pred(\catC)$ embeds non-trivially into a product of categories.
\end{proposition}
\begin{proof}
    We sketch a proof.
    Let $A\in\catC$. For every $p\in\Pred(A)$ we have $p = \truth_I\circ p = (s\ovee s^\perp)\circ p = s\circ p \ovee s^\perp\circ p \equiv p_1\ovee p_2$ and hence we can write $\Pred(A) \cong A_1\oplus A_2 \equiv s\cdot \Pred(A)\oplus s^\perp \cdot \Pred(A)$.
    Let $f\colon A\rightarrow B$ be a morphism in $\catC$. Under the functor $\Pred$ this becomes $\Pred(f)\colon \Pred(B)\rightarrow \Pred(A)$. Using the idempotence of $s$ we note that 
    \begin{align*}
    \Pred(f)(p) &\ =\  \Pred(f)(s\circ p \ovee s^\perp \circ p) \\
    &\ =\ \Pred(f)(s\circ s\circ p \ovee s^\perp\circ s^\perp\circ p) \\
    &\ =\  s\circ \Pred(f)(s\circ p) \ovee s^\perp \circ \Pred(f)(s^\perp\circ p)
    \end{align*}
    and hence $\Pred(f) = (f_1,0) \ovee (0,f_2)$ where $f_i\colon \Pred(B_i)\rightarrow \Pred(A_i)$.
    We then have a faithful functor $\Pred(\catC)\rightarrow \catC_1\times \catC_2$ where $\catC_1$ has objects $s\cdot \Pred(A)$ and morphisms $f\colon s\cdot \Pred(A)\rightarrow s\cdot \Pred(B)$ in $\textbf{EA}$ (or rather the subcategory consisting of effect modules over $s\cdot \Pred(I)$).
    We define $\catC_2$ analogously, but with $s$ and $s^\perp$ interchanged.
\end{proof}

Generally it won't be the case that $\Pred(\catC)$ is equivalent to a product of categories in the above situation. 
However, this is the case if the effectus has some more structure. 
First, let us recall the concept of the \emph{Karoubi envelope}, which allows us to speak of subsystems in a general category, as given by idempotent maps.

\begin{definition}
    Let $\catC$ be any category. We say a morphism $t\colon A\rightarrow A$ in $\catC$ is \Define{idempotent} when $t\circ t = t$.
    We define the \Define{Karoubi envelope} of $\catC$ to be the category $\Split(\catC)$ which has as objects idempotents $t\colon A\rightarrow A$ in $\catC$, and has morphisms $f\colon (t\colon A\rightarrow A)\rightarrow (s\colon B\rightarrow B)$ corresponding to a morphism $f\colon A\rightarrow B$ in $\catC$ satisfying $s\circ f\circ t = f$.
\end{definition}
Note that $\catC$ embeds fully and faithfully into $\Split(\catC)$ via $A\mapsto \id_A$ and $f\mapsto f$. Intuitively we can think of an object $t\colon A\rightarrow A$ as the subobject of $A$ where $t$ holds.
We call this category $\Split(\catC)$ as it makes every idempotent \emph{split}, meaning that for an idempotent $t\colon A\rightarrow A$ we can find a pair of maps $\xi^t\colon A \leftrightarrows A_t\colon \pi_t$ so that $\pi_t\circ \xi^t = t$ and $\xi^t \circ \pi_t = \id_{A_t}$.

\begin{proposition}
    Let $\catC$ be an effectus in partial form.
    Then $\Split(\catC)$ is also an effectus in partial form. Furthermore:
    \begin{itemize}
        \item If $\catC$ is separated by predicates, then so is $\Split(\catC)$.
        \item If $\catC$ is separated by states, then so is $\Split(\catC)$.
        \item If $\catC$ is monoidal, then so is $\Split(\catC)$.
    \end{itemize}
\end{proposition}
\begin{proof}
    We just give a sketch of the necessary constructions to prove this.
    The trivial object of $\Split(\catC)$ is $\id_I$. 
    The predicates of an object $t\colon A\rightarrow A$ in $\Split(\catC)$
    then correspond to maps $p\colon A\rightarrow I$ satisfying $p\circ
    t = p$.
    Note that always~$p \leq \truth_A \after t$,
        but that the converse is not necessarily true.
    The truth element is $\truth_t:= \truth_A\circ t$ and
    falsity is just~$\falsity_t := \falsity_A$.
    We set $p^{\perp_t}:=p^\perp\circ t$.
    It is straightforward to verify that~$\Pred(t)$
    is an effect algebra.
    We define the coproduct as $t+s\colon A+B\rightarrow A+B$ in the obvious way. The coproduct maps then become $\kappa_1' := (t+s)\circ \kappa_1\circ t$ and similarly for $\kappa_2'$.
    All the axioms of an effectus are now easily checked.

    Now suppose $\catC$ is separated by predicates. Let $f,g\colon t\rightarrow s$ in $\Split(\catC)$ where $s\colon B\rightarrow B$ and suppose $p'\circ f = p'\circ g$ for all $p'\colon s\rightarrow \id_I$. We need to show that $f=g$. Let $p\colon B\rightarrow I$ be an arbitrary predicate on $B$. Then $p\circ f = p\circ (s\circ f \circ t) = (p\circ s)\circ (s\circ f \circ t) = (p\circ s)\circ f$. Now since $p\circ s$ is a predicate on $s$, we have $(p\circ s)\circ f = (p\circ s)\circ g$. Doing the argument in reverse we then get $p\circ f = p\circ g$. Now predicate separation in $\catC$ gives $f=g$ as desired. 
    The proof for preservation of separation by states works analogously.

    If $\catC$ is monoidal, we define the monoidal structure in $\Split(\catC)$ to be the same as in $\catC$, with the monoidal unit being $\id_I$. We do need to modify the coherence isomorphisms though. For instance, for an object $t\colon A\rightarrow A$ in $\Split(\catC)$ we set $\lambda_t\colon t\otimes \id_I\rightarrow t$ to be $\lambda_t := t\circ \lambda_A\circ (t\otimes \id_I)$. Using the naturality of $\lambda_A$ it is then straightforward to verify that this satisfies the correct equations. The additional equations required for a monoidal effectus are also easily checked.
\end{proof}

\begin{proposition}\label{prop:monoidal-splits}
    Let $\catC$ be a monoidal effectus with a
    non-trivial idempotent scalar~$s$. Then $\Split(\catC)\cong
    \catC_s\times \catC_{s^\perp}$ for some non-trivial effectuses~$\catC_s$
    and $\catC_{s^\perp}$.
\end{proposition}
\begin{proof}
Without loss of generality we may assume that~$\catC$ has split idempotents (since otherwise we could take $\Split(\catC)$ instead).

    Define $\catC_s$ to be the category with the same objects as~$\catC$,
        but with morphisms restricted to those of the form~$s \cdot f$
        for~$f\colon A \to B$ in~$\catC$
        and with identity~$s\cdot \id_X$ for an object~$A$.
It is easy to see that the coproduct and enrichtment on~$\catC$
        restricts to~$\catC_s$ and hence~$\catC_s$ is a finPAC.
It has the same distinguished object~$I$,
    but for every object $A$ the maximum element of the predicate space~$\catC_s(A,I)$ is $s\cdot \truth_A$,
            and the orthosupplement for $p$ is~$s\cdot p^\perp$.
    Hence~$\catC_s$ is an effectus as well.
    The same holds for~$\catC_{s^\perp}$ with the obvious definition.

We will show~$\catC \cong \catC_s \times \catC_{s^\perp}$. First we need
     a few definitions.
For every object~$A$ in~$\catC$, pick a splitting~$\xi^s\colon A\leftrightarrows A_s\colon \pi^s$
    of the idempotent~$s\cdot \id_A$.
For any morphism~$f\colon A \to B$,
    we define~$f_s\colon A_s \to B_s$ as~$f_s := \xi^s \after f \after \pi_s$.
    Note~$(\id_A)_s = \id_{A_s}$ and $(f \after g)_s = f_s \after g_s$.
Furthermore~$s \cdot \xi^s = \xi^s \after s \cdot \id
        = \xi^s \after \pi_s \after \xi^s  = \xi^s$
        and so~$s \cdot \id_{A_s} = \pi^s \after (s \cdot \xi^s) =  \id_{A_s}$
        and~$s^\perp \cdot \id_{A_s} = 0$.
    Define~$A_{s^\perp}$, $\pi_{s^\perp}$, $\xi^{s^\perp}$  and~$(\ )_{s^\perp}$ similarly.

Now, we can define the functors~$F \colon \catC_s \times \catC_{s^\perp} \leftrightarrows \catC \colon G$
    by
    \begin{equation*}
    F(A,B) \ =\  A_s +B_{s^\perp} \quad F(f, g) \ =\  f_s + g_{s^\perp} \quad
    G(A) \ =\  (A_s, A_{s^\perp}) \quad G(f) \ = \  (f_s, f_{s^\perp}).
    \end{equation*}
To show that these functors form an equivalence of categories, first observe
that~$F(G(A))=A_s+ A_{s^\perp}$.
We have~$\xi^s \after \pi_s = \id$ so that $\truth \after \pi_s = \truth$, and
    hence~$\truth \after \xi^s = s\cdot\truth$
        and similarly~$\truth \after \xi^{s^\perp} = s^\perp\cdot \truth$.
As these are summable predicates, the sum
map~$\langle \xi^s, \xi^{s^\perp} \rangle := (\kappa_1\circ \xi^s)\ovee (\kappa_2\circ \xi^{s^\perp})
    \colon A \to A_s+A_{s^\perp}$ exists, see~\cite[181\textsubscript{VII}]{basthesis}.
    Note that this map is an isomorphism with inverse~$[\pi_s, \pi_{s^\perp}]$.
    From this it follows that $\alpha := [\pi_s, \pi_{s^\perp}]$ is a natural
        isomorphism~$FG \Rightarrow \id$.

For the other direction of the equivalence, first consider the maps
\begin{equation}\label{eq:zkzppp}
    (A_s + B_{s^\perp})_s \ \stackrel[\pi_s]{\xi^s}{\leftrightarrows} \
        A_s + B_{s^\perp} \ \stackrel[\pproj_1]{\kappa_1}{\leftrightarrows}\ 
        A_s \ \stackrel[\pi_s]{\xi^s}{\leftrightarrows} \ 
         A.
\end{equation}
    Note~$\xi^s \after \kappa_1 \after \xi^s \after \pi_s \after
        \pproj_1 \after \pi_s = \xi^s \after (\id_{A_s} + 0_{B_{s^\perp}}) \after \pi_s$.
    Remember~$s \cdot \id_{B_{s^\perp}} = 0$, hence
    \begin{equation*}
    s\cdot (\xi^s \after \kappa_1 \after \xi^s \after \pi_s \after \pproj_1 \after \pi_s)
            \ =\  \xi^s \after (s\cdot \id_{A_s} + s\cdot \id_{B_{s^\perp}}) \after \pi_s
        \ =\  s\cdot \id.
    \end{equation*}
We can also calculate~$
    s\cdot ( \pi_s \after \pproj_1 \after \pi_s \after
    \xi^s \after \kappa_1 \after \xi^s)
    = s \cdot \id$ so that
        the maps in~\eqref{eq:zkzppp} are eachothers inverse
        in~$\catC_s$.
We can do a similar calculation for the map $\pi_{s^\perp} \after \pproj_2 \after \pi_{s^\perp} \colon 
(A_s + B_{s^\perp})_{s^\perp} \rightarrow B$ to show it is an isomorphism in $\catC_{s^\perp}$
so that we see that the map~$\beta = (\pi_s \after \pproj_1 \after \pi_s, \pi_{s^\perp}
    \after \pproj_2 \after \pi_{s^\perp})$ gives a natural isomorphism~$GF \Rightarrow \id$.
\end{proof}

Note that this result does not in fact require all the idempotents to split, just the ones that correspond to scalar multiplication by an idempotent scalar. We can use this fact to get a different variant of this result.

\begin{proposition}\label{prop:monoidal-splits2}
    Let $\catC$ be a monoidal effectus which has images and compatible filters and comprehensions and let $s$ be a non-trivial idempotent scalar.
    Then $\catC \cong \catC_s\times \catC_{s^\perp}$ for some non-trivial effectuses $\catC_s$ and $\catC_{s^\perp}$.
\end{proposition}
\begin{proof}
    We show that the idempotent maps $s\cdot \id_A$ split for every object $A$. The rest of the proof then proceeds as in Proposition~\ref{prop:monoidal-splits}.

    Note first that $s\cdot \truth_A$ is sharp as it is the image of $s\cdot \id_A$. Hence, there is a comprehension $\pi_{s\cdot \truth_A}$ and compatible filter $\xi^{s\cdot \truth_A}$ and we can define $\asrt_{s\cdot \truth_A} := \pi_{s\cdot \truth_A}\after \xi^{s\cdot \truth_A}$.
    We have $\truth\circ (s^\perp\cdot \asrt_{s\cdot \truth_A}) = (s^\perp\cdot \truth)\after \asrt_{s\cdot \truth_A} = \falsity$ so that $s^\perp\cdot \asrt_{s\cdot \truth_A} = 0$. 
    Hence, $s\cdot \asrt_{s\cdot \truth_A} = \asrt_{s\cdot \truth_A}$.
    Then using Lemma~\ref{lem:assert-image} we get $s\cdot \id_A = \asrt_{s\cdot \truth_A}\after (s\cdot \id_A) = s\cdot \asrt_{s\cdot \truth_A} = \asrt_{s\cdot \truth_A}$. So $s\cdot \id_A$ splits via $\pi_{s\cdot \truth_A}\after \xi^{s\cdot \truth_A}$.
\end{proof}

Even without monoidal structure we can get a similar result, but then we have to require separation by either predicates or states.

\begin{proposition}\label{prop:predsep-splits}
    Let $\catC$ be an effectus which is separated by states or predicates and which has images and compatible filters and comprehensions and let $s$ be a non-trivial idempotent scalar.
    Then $\catC \cong \catC_s\times \catC_{s^\perp}$ for some non-trivial effectuses $\catC_s$ and $\catC_{s^\perp}$.
\end{proposition}
\begin{proof}
    This follows along the same lines as the proof of Proposition~\ref{prop:monoidal-splits}, but instead of the maps $s\cdot \id$ we consider the maps $\asrt_{s\after \truth} = \pi_{s\after \truth}\after \xi^{s\after \truth}$.
    We let the category $\catC_s$ have the same objects as $\catC$ but with morphisms $f\colon A\rightarrow B$ satisfying $\im{f}\leq s\after \truth$ and $\truth\after f \leq s\after \truth$. The identity on an object $A$ in $\catC_s$ is $\asrt_{s\after \truth}$. We define $\catC_{s^\perp}$ analogously. 
    As in Proposition~\ref{prop:monoidal-splits} we define $f_s := \xi^{s\after \truth}\after f \after \pi_{s\after \truth}$, and we again get functors~$F : \catC_s \times \catC_{s^\perp} \leftrightarrows \catC : G$. Showing that these form an equivalence of categories follows entirely analogously, except for a complication with establishing that the right-inverse of $\alpha:=[\pi_{s\after\truth},\pi_{s^\perp\after \truth}]$ is $\langle \xi^{s\after\truth}, \xi^{s^\perp\after\truth}\rangle$, which requires predicate separation to be proven.
    We first calculate 
    $[\pi_{s\after\truth},\pi_{s^\perp\after \truth}]\after \langle \xi^{s\after\truth}, \xi^{s^\perp\after\truth}\rangle =
    \asrt_{s\after \truth}\ovee \asrt_{s^\perp\after \truth}$.
    Now first suppose $\catC$ is separated by predicates.
    For any predicate $p$ we have $p\after\asrt_{s\after \truth} = (s\after p)\after \asrt_{s\after \truth} \ovee (s^\perp\after p)\after \asrt_{s\after \truth} = s\after p$ (by Lemma~\ref{lem:assert-image}), so that $p\after (\asrt_{s\after \truth}\ovee \asrt_{s^\perp\after \truth}) = s\after p \ovee s^\perp \after p = p$. 
    As also $p\circ \id = p$ we get by predicate separation $[\pi_{s\after\truth},\pi_{s^\perp\after \truth}]\after \langle \xi^{s\after\truth}, \xi^{s^\perp\after\truth}\rangle = \id$.
    If instead $\catC$ is separated by states, we can do a similar argument, but with the extra complication that we need $s$ to commute with all scalars. This is however true for any idempotent element in an effect monoid~\cite[Lemma~20]{first}.
\end{proof}

These results could perhaps be generalised so that every effectus can be presented as a presheaf from the Boolean algebra of idempotent scalars to the category of effectuses, akin to the results in Ref.~\cite{barbosa2021sheaf}, but as this is not necessary for our results we leave this as future work.

\subsection{Decomposing a directed-complete effectus}

Recall that in a directed-complete effectus, the effect monoid of scalars splits up into a Boolean algebra and a convex effect algebra. We hence get the following.

\begin{proposition}\label{prop:dc-ortho-convex}
    Let $\catC$ be a directed-complete effectus that is separated by states and let $A\in\catC$. Then $\Pred(A) \cong E_1\oplus E_2$ where $E_1$ is an orthoalgebra and $E_2$ is convex (see Definitions~\ref{def:orthoalgebra} and~\ref{def:convex}).
\end{proposition}
\begin{proof}
    Write $M=\Pred(I)$ for the scalars of $\catC$.
    By Corollary~\ref{corscalars}, we know~$M\cong M_1\oplus M_2$ where $M_1$
        is a complete Boolean algebra and $M_2\cong C(X,[0,1])$.
    Let~$s$ be the idempotent that projects onto the Boolean part,
        i.e.~$s \equiv (1,0) \in M_1 \oplus M_2$.
    Each predicate~$p \in E := \Pred(A)$
        splits as~$s \circ p \ovee s^\perp \circ p$.
    Write~$E_1 \equiv s E \equiv \{ s \circ p; p \in E\} $
        and~$E_2 \equiv s^\perp E$.
    As~$s$ is an idempotent,~$sE$ is an effect algebra itself.
    In fact~$\Pred(A)\cong E_1\oplus E_2$.
    We claim that $E_1$ is an orthoalgebra and that $E_2$ is convex.
    The latter statement is easily seen as $M_2$ is convex, and
    hence we can define a convex action on $E_2$ via $\lambda\cdot
    p := (\lambda \cdot s_1)\circ p$.

    To show that $E_1$ is an orthoalgebra we need to prove that
    whenever $p_1\ovee p_1$ is defined for some $p_1\in E_1$
    that then~$p_1 = 0$. So suppose $p_1\ovee p_1$ is defined.
    Then for any state $\omega \colon I\rightarrow A$, the scalar
    $(p_1\ovee p_1)\circ \omega = p_1\circ \omega \ovee
    p_1\circ\omega$ is defined. As $p_1\in E_1$ we have $s_1\circ
    p_1 = p_1$, and hence we have $p_1\circ \omega = s_1\circ
    (p_1\circ \omega)\in M_1$. As $p_1\circ \omega \perp p_1\circ
    \omega$ and $M_1$ is a Boolean algebra, and so in particular
    an orthoalgebra, we must then have $p_1\circ \omega = 0$.
    Now, by separation of states we see that~$p_1 = 0$, as desired.
\end{proof}

Hence, in a directed-complete effectus with separating states, each
predicate space splits up into a sharp part, given by an orthoalgebra,
and a probabilistic part, given by a convex effect algebra. By using
the equivalence of convex effect algebras with vector spaces this
extends to the following.
\begin{proposition}\label{prop:effectus-ortho-OUS}
    Let $\catC$ be a directed-complete effectus that is separated
    by states. $\Pred$ gives a functor $\Pred \colon \catC\rightarrow
    \textbf{OA}^\opp\times \textbf{DCEA}_c^\opp \cong \textbf{OA}^\opp\times
    \textbf{DCOUS}^\opp$, where $\textbf{OA}$ is the category of orthoalgebras and $\textbf{DCOUS}$ is the category of directed-complete order unit spaces.
\end{proposition}

For this result this factorization is merely `internal'
    to the predicate spaces, and not necessarily reflected
    in the structure of the objects of $\catC$,
    as the axioms of a plain effectus do not force
    `predicate sub-spaces' to correspond to objects in the category.
    However, when we have suitable filters and comprehensions 
    we get a stronger result.

\begin{proposition}\label{prop:splits-directed-complete}
    Let $\catC$ be a directed-complete effectus that is separated
    by states and which has images and compatible filters and comprehensions. 
    Then $\catC$ is equivalent to a product of effectuses $\catC\cong \catC_1\times \catC_2$
    and $\Pred$ gives functors $\Pred \colon \catC_1\rightarrow \textbf{OA}^\opp$ 
    and $\Pred \colon \catC_2\rightarrow \textbf{DCEA}_c^\opp \cong \textbf{DCOUS}^\opp$, 
    where $\textbf{OA}$ is the category of orthoalgebras and $\textbf{DCOUS}$ is the category of directed-complete order unit spaces.
\end{proposition}
\begin{proof}
    Let $s$ be the idempotent scalar that splits the scalars of $\catC$ into a Boolean algebra and a convex set.
    Then apply Proposition~\ref{prop:predsep-splits} to get $\catC\cong \catC_1\times \catC_2$ where the scalars of $\catC_1$ are a Boolean algebra and those of $\catC_2$ are convex. Then Proposition~\ref{prop:dc-ortho-convex} shows that the predicate spaces of $\catC_1$ are orthoalgebras and that of $\catC_2$ are convex effect algebras.
    To go from directed-complete convex effect algebras to order unit spaces, apply Proposition~\ref{prop:convex-OUS-equiv}.
\end{proof}

\subsection{Decomposing normal sequential effect algebras}\label{sec:decomposing-SEA}

In our reconstruction we will look at effectuses whose predicate spaces are normal sequential effect algebras (Definition~\ref{defn:sea}). By the previous results, if the effectus has state separation we can consider normal SEAs that are either orthoalgebras or are convex. In the former case, the situation simplifies
    even further.

\begin{lemma}\label{lem:ortho-SEA-boolean}
    Let $E$ be an orthoalgebra that is also a sequential effect algebra.
    Then $E$ is a Boolean algebra with~$a\mult b = a\wedge b$.
\end{lemma}
\begin{proof}
    For any~$a\in E$ we note that~$a\mult a^\perp$ is summable
    with itself as $1 = 1 \mult 1 = (a\ovee a^\perp) \mult
    (a\ovee a^\perp) \geq 2(a\mult a^\perp)$. Since~$E$ is an
    orthoalgebra we must then have~$a\mult a^\perp = 0$ so
    that~$a$ is an idempotent (as $a = a\mult 1 = a\mult (a\ovee a^\perp) = a^2 \ovee a\mult a^\perp = a^2$). But as every element is then an
    idempotent, $E$ must be a Boolean algebra~\cite[Prop.~45]{second}.
\end{proof}

\begin{proposition}\label{prop:boolean-and-convex}
    Let $\catC$ be a directed-complete effectus with separating states, and let $A\in \catC$. If $\Pred(A)$ is a normal SEA, then $\Pred(A)\cong A_b\oplus A_c$ where $A_b$ is a complete Boolean algebra and $A_c$ is a convex effect algebra, or more specifically, the unit interval of a directed-complete order unit space.%
    \footnote{In~\cite{second} it was shown that any normal SEA
    splits up into three parts: a Boolean algebra, a convex effect algebra
    and a type of effect algebra called `purely almost convex'.
    The previous proposition shows that the `purely almost-convex'
    case does not occur in the context of an effectus. Hence,
    this pathological type of SEA is prevented from existing
    in the more compositional setting of an effectus.}
\end{proposition}
\begin{proof}
    By Proposition~\ref{prop:dc-ortho-convex} we see that $\Pred(A)\cong
    A_b\oplus A_c$ where $A_b$ is an orthoalgebra and $A_c$ is
    convex.
    As~$A_b$ is a principal downset and~$p \mult q \leq p$,
        we see that~$A_b$ is also a sequential effect algebra
        and so it must then be a Boolean algebra by
    Lemma~\ref{lem:ortho-SEA-boolean}.
\end{proof}

\subsection{Decomposing finite-dimensional effectuses}

Let us briefly demonstrate how we can talk about finite-dimensional systems in effectus theory as
the literature on generalised probabilistic theories often restricts to working with finite-dimensional spaces.
This is motivated by the operational assumption of \emph{finite tomography}.
\begin{definition}\label{def:finite-tomography}
    We say an effectus has \Define{finite tomography} when for each object $A$ there is a finite set of predicates $p_1,\ldots, p_k$ such that for any pair of morphisms $f,g \colon B\rightarrow A$ we have $f=g$ iff $p_i\circ f=p_i\circ g$ for all $i$.
\end{definition}
The reasoning behind this is that we can physically only probe a state transformation through the application of a predicate, and we can only ever do this a finite number of times. Hence, finite tomography says that there is such a finite set of probings that is sufficient to fully determine a transformation.

\begin{proposition}
    Let $\catC$ be a directed-complete effectus which has finite tomography.
    Then there exists a finite set $A$ and $n\in \N$ such that
    $\Pred(I) \cong \mathcal{P}(A)\oplus [0,1]^n$, where $\mathcal{P}(A)$ denotes the powerset of $A$.
\end{proposition}
\begin{proof}
    By finite tomography we know that there is a finite collection $p_1,\ldots p_k\in \Pred(I)$ such that $a,b\in \Pred(I)$ are equal iff $p_i\cdot a = p_i\cdot b$ for all $i$.
    As $\catC$ is a directed complete effectus we know that $\Pred(I)\cong B\oplus C(X,[0,1])$ for some complete Boolean algebra $B$ and extremally-disconnected compact Hausdorff space $X$.
    Let $s$ denote the idempotent scalar splitting these two parts.
    Then $\{s\cdot p_i\}$ is a set of elements of $B$ that
    suffice to separate the elements of $B$. This is only
        possible if $B$ is finite, hence~$\mathcal{P}(A) \cong B$ for some finite set~$A$.

    Similarly, the set~$\{s^\perp\cdot p_i\}$ separates~$C(X,[0,1])$
        and in fact~$C(X)$ itself.
        This is only possible if~$C(X)$ is finite dimensional, viz.~if~$X$ is finite (indeed, the $s^\perp\cdot p_i$ can only distinguish elements that lie in their span, and hence $C(X)$ can have dimension at most $k$).
    The only finite compact Hausdorff spaces are discrete and hence $C(X)\cong \R^n$ for some $n$.
\end{proof}

As $\mathcal{P}(A)$ is equivalent to a product of the trivial Boolean algebra $\{0,1\}$ and $[0,1]^n$ is a product of the `trivial' convex effect algebra $[0,1]$, the above result means that, by iterated application of Proposition~\ref{prop:effectus-product}, for a directed complete effectus with finite tomography $\catC$, the category $\Pred(\catC)$ embeds into a product of categories $B_1\times\cdots\times B_k\times C_1\times\cdots\times C_k$ where each $B_i$ is a category with `scalars' $\{0,1\}$ and each $C_i$ is a category with `scalars' $[0,1]$.
If the effectus additionally has images, compatible filters and comprehensions and state or predicate separation then we can apply Proposition~\ref{prop:predsep-splits} to see that the effectus is equivalent to a product of effectuses where the scalars are either $\{0,1\}$ or $[0,1]$.
For each $C_i$ we can also show that the predicate spaces correspond to finite-dimensional vector spaces using finite tomography, and hence we have retrieved the standard framework of generalised probabilistic theories.

\section{The reconstruction}\label{sec:reconstruction}

We are ready to state our definition of a type of effectus that will lead us to quantum theory.

\begin{definition}\label{def:sequential-effectus}
  A \Define{sequential} effectus is a normal effectus separated by states satisfying
    the following.
  \begin{enumerate}
    \item The effectus has filters and comprehensions.
    \item Comprehensions have images.
    \item The pure maps form a dagger category. 
    \item Every pure map $f$ is $\diamond$-adjoint to $f^\dagger$.
    \item For every predicate $p\in \Pred(A)$ there is a unique $\dagger$-positive pure map $\asrt_p\colon A\rightarrow A$ satisfying $\truth_A\circ \asrt_p = p$ called the \Define{assert map} of $p$.
    \item For every object $A$, the operation $\& \colon \Pred(A)\times \Pred(A)\rightarrow \Pred(A)$ given by $p\& q := q\circ \asrt_p$ is a normal sequential product, making $\Pred(A)$ into a normal SEA.
  \end{enumerate}
\end{definition}

\begin{remark}
    Points 1--4 are variations on the properties outlined in Remark~\ref{rem:reconstruction-finite}. Point 5 is a new assumption. We remark that the uniqueness condition of point 5 can be framed as the implication $\truth\circ f^\dagger\circ f = \truth\circ g^\dagger\circ g \implies f^\dagger\circ f = g^\dagger\circ g$ for any pure $f$ and $g$. In this sense it is similar to the \emph{CPM axiom}~\cite{coecke2010environment} for an \emph{environment structure} (\ie~choice of pure maps) which states that $\truth\circ f = \truth\circ g \iff f^\dagger\circ f = g^\dagger\circ g$. We will see that for sharp predicates this definition of an assert map reduces to that of Definition~\ref{def:sharp-assert}.
\end{remark}

The goal of this section will be to prove the following pair of theorems.

\begin{theorem}\label{thm:JB-embedding}
  Let $\catC$ be a sequential effectus. Then $\catC$ is equivalent to a product of effectuses $\catC\cong \catC_1\times \catC_2$ where the predicate spaces of $\catC_1$ are complete Boolean algebras and those of $\catC_2$ are directed-complete JB-algebras. 
  In particular we have predicate functors $\Pred\colon \catC_1\rightarrow \textbf{CBA}^\opp$ and $\Pred\colon \catC_2\rightarrow \textbf{JB}_{\text{npc}}^\opp$, where $\textbf{CBA}$ denotes the category of complete Boolean algebras and monotone maps that are faithful iff $\catC$ is separated by predicates.
\end{theorem}

\begin{theorem}\label{thm:JBW-CBA}
    Let $\catC$ be a sequential effectus with irreducible scalars. Then all predicate spaces are either complete Boolean algebras or they are all the unit interval of a JBW-algebra. Furthermore, there is a functor $\Pred\colon\catC\rightarrow \textbf{D}^\opp$ where $\textbf{D}$ is either $\textbf{CBA}$ of $\textbf{JBW}_{\text{npc}}$, and this functor is faithful iff $\catC$ is separated by predicates.
\end{theorem}

\begin{remark}
In the first theorem we talk about directed-complete JB-algebras, while in the second we talk about JBW-algebras, i.e.~those directed-complete JB-algebras that are separated by normal states. The difference comes from the fact that in the first case our scalars can satisfy $\Pred(I)\cong [0,1]_{C(X)}$ where $X$ is an arbitrary Stonean space, and hence $C(X)$ does not have to be a JBW-algebra. This might seem surprising as in Definition~\ref{def:sequential-effectus} we require separation of states and that all maps are normal. However, this talks about states as maps $I\rightarrow A$ in the effectus, while the condition to be a JBW-algebra requires normal states as maps $V_A\rightarrow \R$, where $V_A$ is the order unit space corresponding to $A$. For instance, to satisfy the condition of separation by states on $I$ we only need the map $\id\colon I\rightarrow I$. It is unclear which, categorically nice, condition we could require that forces systems to be separated by normal states in the correct sense, without restricting the scalars to the irreducible ones.
\end{remark}

For the remainder of the section, we will let $\catC$ be a sequential effectus, and we let $A$ denote an arbitrary system in $\catC$.
By assumption $\Pred(A)$ is a normal SEA.
By Proposition~\ref{prop:boolean-and-convex}, it is a direct sum of a complete Boolean algebra and a convex SEA. 
We aim to use a version of Theorem~\ref{thm:normalSEAisJB} to show the convex part is the unit interval of a JB-algebra.
This means that we need to show that the sequential product given by the assert maps is compressible and quadratic (up to some technical modifications that will become apparant later on).

\subsection{Sequential effectuses are \texorpdfstring{$\diamond$}{diamond}-effectuses}

Our definition refers to $\diamond$-adjointness, but this concept is only well-behaved when the effectus is a $\diamond$-effectus. So let us start by showing
that a sequential effectus is indeed a $\diamond$-effectus. This  means we still need to show that images of all maps exist and that $p$ is sharp iff $p^\perp$ is sharp.
For any predicate $p$ we will write $p^2 := p\mult p := p\circ \asrt_p$. Following Definition~\ref{defn:sea} we call predicates with $p^2=p$ \Define{idempotent}.

\begin{proposition}\label{prop:idempotent-is-sharp}
	Let $p\in \Pred(A)$ be any predicate. Then the image of $\asrt_p$ exists and is equal to $\ceil{p}$, where $\ceil{p}$ is both as in Lemma~\ref{lem:normal-SEA-properties} as well as in Definition~\ref{def:floorceiling}. 
    In particular, $\ceil{p}$ is sharp, $p$ is sharp iff $p^\perp$ is sharp, and $p$ is sharp iff it is idempotent.
\end{proposition}
\begin{proof}
	Let $q\in\Pred(A)$ be a predicate with $\truth\circ \asrt_p = q\circ \asrt_p$. I.e.~$p=q\mult p$. By Lemma~\ref{lem:normal-SEA-properties}.c) we then have $q\geq \ceil{p}$. As $\ceil{p}\mult p = p$ we see that indeed $\im{\asrt_p} = \ceil{p}$.
    In particular the ceiling given in Lemma~\ref{lem:normal-SEA-properties}.c) is a sharp predicate.
    As $\asrt_p$ is $\dagger$-self-adjoint, and thus $\diamond$-self-adjoint we also calculate 
    $\im{\asrt_p} = (\asrt_p)_\diamond(\truth) = \asrt_p^\diamond(\truth) = \ceil{\truth\after\asrt_p} = \ceil{p}$, where here $\ceil{p}$ is as in Definition~\ref{def:floorceiling}, so that the two possible definitions of the ceiling coincide, and in particular $\ceil{p}$ is sharp.
    Now suppose $p$ is sharp, so that $\floor{p} = p$. Then $p^\perp = \floor{p}^\perp = \ceil{p^\perp}$ is sharp. So indeed $p$ sharp iff $p^\perp$ sharp.

    Suppose again that $p$ is sharp. Then $\ceil{p} = p$, and hence $\im{\asrt_p} = p$ so that $p^2 := p\circ\asrt_p = \truth\circ\asrt_p = p$. So $p$ is idempotent. Conversely, if $p^2 = p$, then $p\geq \im{\asrt_p} = \ceil{p} \geq p$ so that $p=\ceil{p}$ and $p$ is sharp.
\end{proof}


\begin{proposition}\label{prop:sequential-has-images}
	All morphisms have images.
\end{proposition}
\begin{proof}
	Let $f$ be a morphism in our category.
	The idempotents of a normal SEA form a complete lattice~\cite{second}. Hence, define $\im{f} := \bigwedge \{p^2=p~;~p\circ f = \truth\circ f\}$. By normality of $f$ we have $\im{f}\circ f = \truth\circ f$. Now, if $q\circ f = \truth\circ f$, then also $\floor{q}\circ f = \truth\circ f$ (follows from Proposition~\ref{prop:floorceiling}.e)). As $\floor{q}$ is sharp, it is idempotent, and hence $\im{f}\leq \floor{q}\leq q$ as desired.
\end{proof}

\begin{corollary}
  A sequential effectus is a $\diamond$-effectus.
\end{corollary}

\subsection{The sequential product is compressible}

Now we will venture to prove that the sequential product is compressible (cf.~Definition~\ref{def:SEA-compressible}). To do this we need to know more about the relationship between the assert maps, filters, comprehensions and the dagger.

\begin{lemma}
  For all predicates $p\in \Pred(A)$ we have
    $\asrt_p^2 = \asrt_{p^2}$.
\end{lemma}
\begin{proof}
    By definition $\asrt_p$ is $\dagger$-positive, and so in particular is $\dagger$-self-adjoint. Hence, $\asrt_p^2 = \asrt_p^\dagger\circ \asrt_p$ is also $\dagger$-positive. We have $\truth\circ \asrt_p^2 = p\circ \asrt_p = p^2 = \truth\circ \asrt_{p^2}$ so that by the uniqueness of $\dagger$-positive maps: $\asrt_p^2 = \asrt_{p^2}$.
\end{proof}



\begin{lemma}
\label{lem:compatible-filter-compression}
  Comprehensions and filters are compatible (cf.~Definition~\ref{def:compatible-quotients-comprehensions}), i.e.~for any comprehension $\pi_p$ of a sharp predicate $p$ there exists a filter $\xi^p$ of $p$ such that $\xi^p \circ \pi_p = \id$.
  Furthermore, $\pi_p\circ \xi^p = \asrt_p$.
\end{lemma}
\begin{proof}
  Let $p$ be a sharp predicate, which is hence idempotent. 
  As $\asrt_p$ is pure we have $\asrt_p = \pi\circ\xi$ for some comprehension $\pi$ and filter $\xi$. 
  Now, 
  $$\pi\circ \xi 
  \ =\  \asrt_p 
  \ =\  \asrt_{p^2} \ =\  \asrt_p\circ\asrt_p 
  \ = \  \pi\circ \xi\circ\pi\circ \xi$$
  so that $\xi\circ\pi = \id$ 
  (as filters are epic and comprehensions are monic). 
  As $\truth\circ\pi = \truth$ we calculate
  $\truth\circ\xi = \truth\circ\pi\circ\xi = \truth\circ\asrt_p = p$ 
  so that $\xi$ is a filter for $p$. 
  Furthermore, $\im{\pi} = p$ as
  $\truth\circ \pi = \truth\circ\pi\circ\xi\circ\pi = \truth\circ\asrt_p\circ \pi = p\circ\pi$
  and if $q\circ \pi = \truth\circ\pi$ then $p = \truth\circ \asrt_p = \truth\circ\pi\circ\xi = q\circ\pi\circ\xi = q\circ \asrt_p = q\mult p$, so that $p\geq \ceil{q}\geq q$ by Lemma~\ref{lem:normal-SEA-properties}.c).
  Hence $\pi$ is a comprehension for $p$.
  Now let $\pi'$ be another comprehension for $p$. Then $\pi' = \pi\circ\Theta$ for some isomorphism $\Theta$. 
  Define $\xi' := \Theta^{-1}\circ\xi$. 
  Then $\pi'\circ\xi' = \pi\circ\xi = \asrt_p$ and $\xi'\circ\pi' = \id$.
\end{proof}

\begin{proposition}
\label{prop:dagger-of-compression}
    Let $p$ be a sharp predicate. Then $\pi_p^\dagger = \xi^p$, where $\pi_p$ and $\xi^p$ form a pair of a comprehension and a filter of $p$ with $\xi^p\circ \pi_p = \id$ and $\pi_p\circ\xi^p = \asrt_p$.
\end{proposition}
\begin{proof}
    Let $\pi_p$ be a comprehension of a sharp predicate $p$, and let $\xi^p$ be a filter of $p$ such that $\xi^p\circ\pi_p = \id$ which exists by Lemma~\ref{lem:compatible-filter-compression}.
    By Definition~\ref{def:sequential-effectus}, $\pi_p$ is $\diamond$-adjoint to $\pi_p^\dagger$. 
    Hence 
    $\ceil{\truth\circ \pi_p^\dagger} 
    = (\pi_p^\dagger)^\diamond(\truth) 
    = (\pi_p)_\diamond(\truth) 
    = \im{\pi_p} = p$. 
    Similarly, we calculate $\im{(\xi^p)^\dagger} = p$. 
    By the universal property of filters respectively comprehensions
    there are then unique maps $h$ and $g$ such that 
    $\pi_p^\dagger = h\circ\xi^p$ and 
    $(\xi^p)^\dagger = \pi_p\circ g$. 
    Using $\xi^p\circ\pi_p = \id$ twice we calculate 
    $\id = \id^\dagger 
    = \pi_p^\dagger \circ (\xi^p)^\dagger 
    = h\circ\xi^p \circ \pi_p \circ g 
    = h\circ g$. 
    As a result $\truth=\truth\circ \id = \truth\circ h\circ g \leq \truth\circ g$ so that $g$ is unital, 
    and hence $\truth\circ (\xi^p)^\dagger = \truth\circ\pi_p\circ g = \truth$.

  By uniqueness of $\dagger$-positive maps we have $(\xi^p)^\dagger\circ \xi^p = \asrt_{\truth\circ (\xi^p)^\dagger\circ \xi^p} = \asrt_{\truth\circ \xi^p} = \asrt_p = \pi_p\circ \xi^p$. 
  Because $\xi^p$ is epic we conclude that indeed $(\xi^p)^\dagger = \pi_p$.
\end{proof}

As a consequence of this proposition we note that $\pi_p^\dagger\circ \pi_p = \id$ and $\pi_p\circ\pi_p^\dagger = \asrt_p$ for any sharp $p$. This makes the comprehensions into \emph{dagger-kernels}~\cite{heunen2010quantum}, and furthermore, we have now recovered the conditions specified in Remark~\ref{rem:reconstruction-finite}. We also have the following corollary.

\begin{corollary}\label{cor:dagger-of-iso}
    Let $\Theta$ be an isomorphism. Then $\Theta^\dagger = \Theta^{-1}$.
\end{corollary}
\begin{proof}
    $\Theta$ is a filter for $\truth$, which is sharp, 
    and $\Theta^{-1}$ is a comprehension for $\truth$. 
    As $\Theta^{-1}\circ \Theta = \id$ and $\Theta\circ \Theta^{-1} = \id = \asrt_{\truth}$, they satisfy the conditions of the previous proposition.
\end{proof}

\begin{proposition}\label{prop:compressive}
    The sequential product is compressible (cf.~Definition~\ref{def:SEA-compressible}).
\end{proposition}
\begin{proof}
Let $p$ be sharp and let $\omega$ be a state such that $p\circ \omega = 1$. 
We need to show that $\asrt_p\circ\omega = \omega$ as then $(p\mult a)\circ \omega := a\circ\asrt_p\circ\omega = a\circ \omega$ as desired. But as $p\circ \omega = 1$ implies $\im{\omega}\leq p$, this follows immediately from Lemma~\ref{lem:assert-image}.a).
\end{proof}

\subsection{The sequential product is quadratic}

That the sequential product is also quadratic (cf.~Defintion~\ref{def:SEA-quadratic}) requires a bit more work.

\begin{lemma}
\label{lem:pure-decomposition}
    Every pure map $f$ factors as 
    $f=\pi_{\im{f}}\circ \Theta\circ \xi^{\ceil{\truth\circ f}}\circ \asrt_{\truth\circ f}$ 
    where $\Theta$ is an isomorphism.
\end{lemma}
\begin{proof}
	As $f$ is pure, it is by definition of the form $f=\pi\circ \xi$ for some comprehension $\pi$ and filter $\xi$. It is then straightforward to show that in fact $f=\pi_{\im{f}}\circ\Theta\circ \xi^{\truth\circ f}$, where $\Theta$ is an isomorphism.
  It hence remains to show that $\xi^{\ceil{\truth\circ f}}\circ \asrt_{\truth\circ f}$ is a filter for $\truth\circ f$. 
  Write $p:=\truth\circ f$. 
  Note first of all that 
  $$\truth\circ (\xi^{\ceil{p}}\circ \asrt_p) 
  \ =\  \ceil{p}\circ\asrt_p 
  \ =\  \truth\circ \asrt_{\ceil{p}}\circ \asrt_p 
  \ \stackrel{\ref{lem:assert-image}.a)}{=}\ \truth\circ \asrt_p \ =\  p$$
  so that it remains to show that $\xi^{\ceil{p}}\circ \asrt_p$ is a filter. 
  We see that 
  \begin{align*}
  \im(\xi^{\ceil{p}}\circ \asrt_p)
  &\ =\ (\xi^{\ceil{p}}\circ \asrt_p)_\diamond(\truth) \\
  &\ =\  (\xi^{\ceil{p}})_\diamond((\asrt_p)_\diamond(\truth)) \\
  &\ =\  (\xi^{\ceil{p}})_\diamond(\ceil{p})  \\
  &\ =\  \im(\xi^{\ceil{p}}\circ\pi_{\ceil{p}}) \\
  &\ =\  \im \id  \\
  &\ =\  \truth.
  \end{align*}
  Being a composition of pure maps, 
  $\xi^{\ceil{p}}\circ \asrt_p$ is a pure map itself, 
  and hence is equal to $\pi\circ\xi$ for some comprehension $\pi$ and filter $\xi$. Now we calculate 
  $\truth=\im(\xi^{\ceil{p}}\circ \asrt_p) = \im(\pi\circ\xi) \leq \im{\pi}$ 
  so that $\im{\pi}=\truth$ and hence $\pi$ is an isomorphism. 
  We conclude that $\xi^{\ceil{p}}\circ \asrt_p$ is a filter.
\end{proof}

\begin{proposition}
    Let $p$ and $q$ be arbitrary predicates on the same object. Then 
    $$\asrt_{p\& q}^2 \ =\  \asrt_p \circ \asrt_q^2 \circ \asrt_p.$$
\end{proposition}
\begin{proof}
  First we note that for any assert map $(\asrt_p)^\diamond = (\asrt_p)_\diamond$, as assert maps are $\dagger$-self-adjoint.
  Now we calculate $\truth\circ \asrt_q\circ\asrt_p  =  p\mult q$ and 
  \begin{align*}
  \im{(\asrt_q\circ \asrt_p)} &\ =\  (\asrt_q\circ\asrt_p)_\diamond(\truth) \\
  &\ =\  (\asrt_q)_\diamond\circ(\asrt_p)_\diamond(\truth) \\
  &\ =\  (\asrt_q)_\diamond(\ceil{p}) \\
  &\ =\  (\asrt_q)^\diamond(\ceil{p}) \\
  &\ =\  \ceil{\ceil{p}\after\asrt_q} \\
  &\ =\  \ceil{p\after \asrt_q},
  \end{align*}
  where the last step follows from Proposition~\ref{prop:floorceiling}.d).
  Write $p\after\asrt_q = q\mult p$
  and use Lemma~\ref{lem:pure-decomposition} to get 
  $$\asrt_q\circ\asrt_p \ =\  \pi_{\ceil{q\mult p}}\circ \Theta\circ\xi^{\ceil{p\mult q}}\circ \asrt_{p\mult q}$$
   for some isomorphism $\Theta$. Applying the dagger to both sides and using Proposition~\ref{prop:dagger-of-compression} and Corollary~\ref{cor:dagger-of-iso} gives us:
   $$
   \asrt_p\circ\asrt_q = (\asrt_q\circ \asrt_p)^\dagger 
   = \asrt_{p\mult q}\circ (\xi^{\ceil{p\mult q}})^\dagger \circ \Theta^\dagger \circ \pi_{\ceil{q\mult p}}^\dagger
   = \asrt_{p\mult q}\circ \pi_{\ceil{p\mult q}} \circ \Theta^{-1} \circ \xi^{\ceil{q \mult p}}.
   $$
   Finally, we calculate:
   \begin{align*}
    \asrt_p\circ\asrt_q^2\circ \asrt_p
    &\ =\ \asrt_{p\mult q}\circ \pi_{\ceil{p\mult q}} \circ \Theta^{-1} \circ \xi^{\ceil{q\mult p}}\circ \pi_{\ceil{q\mult p}}\circ \Theta\circ\xi^{\ceil{p\mult q}}\circ \asrt_{p\mult q} \\
    &\ =\  \asrt_{p\mult q}\circ \pi_{\ceil{p\mult q}}\circ\xi^{\ceil{p\mult q}}\circ \asrt_{p\mult q} \\
    &\ =\  \asrt_{p\mult q}\circ \asrt_{\ceil{p\mult q}}\circ \asrt_{p\mult q} \\
    &\ =\ \asrt_{p\mult q}\circ \asrt_{p\mult q} \\
    &\ =\ \asrt_{(p\mult q)^2}.
   \end{align*}
   And hence we are done.
\end{proof}

Now we can conclude that the sequential product in $\Pred(A)$ is quadratic (Dfn.~\ref{def:SEA-quadratic}).

\begin{corollary}\label{cor:assert-is-quadratic}
    Let $p$ and $q$ be sharp predicates. Then $(p\& q)^2 = p\&(q\& p)$.
\end{corollary}
\begin{proof}
    Just plug $\truth$ into the expression of the previous proposition and use $\asrt_q^2 = \asrt_q$ for sharp $q$.
\end{proof}

\begin{remark}
    An interesting question to ask is whether the dagger structure we impose on the pure maps is structure or a property, i.e.~whether it is unique given the other assumptions we impose on it. 
    Lemma~\ref{lem:pure-decomposition} shows that each pure map decomposes into an assert map, filter for a sharp predicate, isomorphism, and comprehension for a sharp predicate, so that it suffices to consider the uniqueness of the dagger for these four classes of maps. Using Proposition~\ref{prop:dagger-of-compression} and Corollary~\ref{cor:dagger-of-iso} it is relatively straightforward to show the dagger is uniquely defined for isomorphisms and filters and comprehensions for sharp predicates. This leaves the question as to whether the assert maps (for non-sharp predicates) are independent of the chosen dagger. We require the assert maps to lead to a normal sequential product. Uniqueness of these products is analysed in~\cite{wetering2018characterisation}. In particular, by applying~\cite[Theorem~V.19]{wetering2018characterisation} we can show that our assert maps are unique when the predicate is \emph{simple} (a finite linear combination of sharp predicates) and \emph{invertible}. These predicates form a norm-dense set of the predicates. Without any further conditions it is however not clear whether the assert maps are unique for non-simple predicates. However, our assert maps are also $\diamond$-positive, and we will later see that if we impose that our scalars are irreducible that our predicate spaces are JBW-algebras (Theorem~\ref{thm:JBW-CBA}). In~\cite[Theorem~4.6.17]{vandewetering2021thesis} it is shown that $\diamond$-positivity uniquely determines a map by its action on the unit in a JBW-algebra, so that in this setting the assert maps are uniquely determined.
\end{remark}

\subsection{The predicate spaces are JB-algebras}

By assumption our effectus $\catC$ has filters and comprehensions and is separated by states. By Proposition~\ref{prop:sequential-has-images} it has images, and by Lemma~\ref{lem:compatible-filter-compression} the filters and comprehensions are compatible.
Hence, Proposition~\ref{prop:splits-directed-complete} applies and $\catC\cong \catC_1\times \catC_2$ where $\catC_1$ has Boolean scalars and $\catC_2$ has convex scalars. The predicate spaces of $\catC_1$ are then complete Boolean algebras, so we are done with that part. We may then focus on $\catC_2$. So without loss of generality assume that $\catC$ has convex scalars, so that all predicate spaces are convex normal SEAs. Let $A$ again denote an object of $\catC$ and let $V_A$ denote the order unit space such that $\Pred(A)\cong [0,1]_{V_A}$.

By the previous results the sequential product is compressible and quadratic, so it looks like we can now use Theorem~\ref{thm:normalSEAisJB} to finish our proof that the predicate spaces are JB-algebras. However, a subtle issue now arises. The notion of being compressible as defined in Definition~\ref{def:SEA-compressible} refers to states $V_A\rightarrow \R$, while our notion of state internal to the effectus is a map $V_A\rightarrow C(X)$ where $C(X) =: V_I$ for the trivial object $I$. Furthermore, Definition~\ref{def:SEA-compressible} requires the property to hold for \emph{all} states, while here we have only shown it to hold for states internal to the category.
However, it is still possible to get to the conclusion of Theorem~\ref{thm:normalSEAisJB}. To do so, we have to delve into the details behind Theorem~\ref{thm:normalSEAisJB}.

The crucial part of the proof of Theorem~\ref{thm:normalSEAisJB} is to show that the operators $D_p:=\asrt_p-\asrt_{p^\perp}$ for sharp predicates $p\in\Pred(A)$, when viewed as acting on the order unit space $V_A$, are \emph{order derivations}.
\begin{definition}
    Let $W$ be an order unit space, and let~$\delta \colon W\rightarrow W$ be a bounded linear map. We call $\delta$ an \Define{order derivation} when $e^{t\delta} := \sum_{n=0}^\infty \frac{(t\delta)^n}{n!}$ is an order isomorphism for all $t\in \R$.
\end{definition}

A useful way to prove a map is an order derivation is to use the following proposition.
\begin{proposition}[{\cite[Proposition 1.108]{alfsen2012state}}]
\label{prop:order-derivation-condition}
    Let $W$ be a Banach order unit space, and let~$\delta\colon W\rightarrow W$ be a bounded linear map. Then $\delta$ is an order derivation if and only if for all $a\in W^+$ and states $\omega\colon W\rightarrow \R$ the following implication holds: 
    $$\omega(a) = 0 \ \implies\  \omega(\delta (a)) = 0.$$
\end{proposition}

We are interested in taking $\delta:=D_p:=\asrt_p-\asrt_{p^\perp}$ for some sharp predicate $p$. The condition $\omega(a) = 0 \implies \omega(\delta(a)) = 0$ then becomes equivalent to the implication~$\omega(a) = 0 \implies \omega(p\mult a) = \omega(p^\perp\mult a)$. In~\cite{wetering2018sequential} it is shown that this implication follows when the SEA is compressive and quadratic. We have shown these conditions, but only for the states \emph{internal} to the effectus. So we have the following.
\begin{lemma}\label{lem:state-order-lemma}
    Let $p\in\Pred(A)$ be sharp and $a\in\Pred(A)$ arbitrary. Let $\omega\colon I\rightarrow A$ be any state on $A$. Then $a\circ\omega = 0 \implies (p\mult a)\circ \omega = (p^\perp\mult a)\circ \omega$.
\end{lemma}
\begin{proof}
    This follows in exactly the same way as for compressive quadratic SEAs as shown in Proposition~46 of~\cite{wetering2018sequential}.
\end{proof}

As our internal states do not necessarily correspond to the type of states mentioned in Proposition~\ref{prop:order-derivation-condition}, we need to rework the proof of that statement to make it apply in our situation.
In particular, we need to modify the proof of Theorem~1.106 in~\cite{alfsen2012state} on which Proposition~\ref{prop:order-derivation-condition} depends.

\begin{proposition}\label{prop:assert-is-derivation}
    Let $p\in\Pred(A)$ be a sharp predicate. Then $D_p:=\asrt_p-\asrt_{p^\perp}$ is an order derivation on $V_A$.
\end{proposition}
\begin{proof}
    To mimic the notation of Theorem~1.106 of~\cite{alfsen2012state}, write $\delta:= D_p$.
    We need to show that $e^{t\delta}\geq 0$ for all $t\in\R_{>0}$ (the result for $t<0$ follows by repeating the argument with $-\delta$).
    To do this, it suffices to show that $(1 - \lambda\delta)^{-1}\geq 0$ for all $\lambda < \frac12\norm{\delta}^{-1}$, as we can then calculate
    \[e^{t\delta}
        \ =\ (e^{-t\delta})^{-1}
        \ =\  \Bigl(\lim_n(1-\nicefrac{t}{n}\,\delta)^n\Bigr)^{-1} 
        \ =\  \lim_n \bigl((1-\nicefrac{t}{n}\,\delta)^{-1}\bigr)^n\]
    which then indeed is positive as $\nicefrac{t}{n}<\frac12\norm{\delta}^{-1}$ for sufficiently large $n$.

    To prove $(1-\lambda\delta)^{-1} \geq 0$, we can use Eq.~(1.82) of~\cite{alfsen2012state}, which says this is the case precisely when $1-\lambda\delta$ maps non-positive elements to non-positive elements.

    Hence, let $y\in V_A$ with $y\not\in V_A^+$. We will show that $(1-\lambda\delta)y \not\in V_A^+$ when $\lambda <\frac12\norm{\delta}$. We do this by finding a state $\omega$ such that $\omega((1-\lambda\delta)y) < 0$. As we have $\omega(a)\geq 0$ for all $a\geq 0$ this establishes that $(1-\lambda\delta)y$ is indeed not positive.

    Using the spectral theorem of convex normal SEAs write $y=y^+-y^-$ for $y^+,y^-\geq 0$ and $y^+\mult y^- = 0$. 
    Let $\alpha = \norm{y^-}$ be the `absolute value of the minimal eigenvalue of $y$'.
    Again using the spectral theorem we can then find an idempotent effect $p\neq 0$ which projects onto the part where $y$ is `very negative', i.e.~$p\mult y < -\frac\alpha2 p$.
    By separation of states there is a state $\omega\colon I\rightarrow A$ such that $p\circ \omega\neq 0$.
    We may assume that $\im{\omega}\leq p$, as otherwise we can simply take $\omega':=\asrt_p\circ\omega$. Note that we can equivalently view $\omega$ as a positive linear map $\omega\colon V_A\rightarrow C(X)$ where $C(X)$ corresponds to the predicate space of $I$.

    Now, note that $\omega(y) = \omega(p\mult y) \leq \omega(-\frac\alpha2 p) = -\frac\alpha2\omega(1)$.
    Let $y'\in V_A$ be such that $\norm{y-y'}<\frac\alpha2$. In particular, this means that for any state $\sigma$ we have $\sigma(y-y') < \frac\alpha2\sigma(1)$. Then we also have
    $$\omega(y')
   \  =\  \omega(y'-y+y) 
    \ =\  \omega(y'-y) + \omega(y)
    \ <\ \frac\alpha2\omega(1) - \frac\alpha2\omega(1)\ =\ 0.$$

    As a result, for any $z\in V_A$ with $\norm{z}<1$ we see that $\omega(y+\frac\alpha2 z) < 0$ so that $\omega(z) < -\frac2\alpha\omega(y)$. 
    As this holds for all $z$ with $\norm{z}<1$ we get $\norm{\omega}\leq -\frac2\alpha\omega(y)$.

    Set $x:=y^+$. As $x\mult y^- = 0$, we have $\omega(x) = 0$
    so that by Lemma~\ref{lem:state-order-lemma} we have $\omega(\delta x) = 0$. 
    Note furthermore that $\norm{y-x} = \norm{y^-} = \alpha$.

    Let $\lambda\in\R_{>0}$. We calculate: 
    \begin{align*}
    \omega((1-\lambda\delta)y)
    &\ =\  \omega(y)-\lambda\omega(\delta y) \\
    &\ =\  \omega(y)-\lambda\omega(\delta(y-x)) \\
    &\ \leq\ \omega(y) +\lambda\norm{\omega}\,\norm{\delta}\,\norm{y-x} \\
        &\ \leq\ \omega(y) -\frac{2\lambda}\alpha \omega(y)\norm{\delta}\alpha \\
    &\ =\  (1-2\lambda\norm{\delta})\omega(y).
    \end{align*}
    Hence, if $2\lambda\norm{\delta}<1$ we see~$\omega((1-\lambda \delta)y)<0$.
    As $\omega$ is positive, this means that $(1-\lambda\delta)y\not\in V_A^+$ when $\lambda<\frac12\norm{\delta}^{-1}$. 
    As $y$ was an arbitrary non-positive element, we indeed conclude that $1-\lambda\delta$ carries $V_A\backslash V_A^+$ into itself.
\end{proof}

\begin{proposition}\label{prop:is-JB-algebra}
    The space $V_A$ is a JB-algebra.
\end{proposition}
\begin{proof}
    This can be shown by invoking Theorem~9.48 of~\cite{alfsen2012geometry} (which itself uses Theorem~9.43 of~\cite{alfsen2012geometry}). We can use this Theorem~9.48 because of Proposition~\ref{prop:assert-is-derivation}. We give a brief sketch of the proof.

    Let $p\in\Pred(A)$ be sharp. We claim that the operator $T_p:= \frac12 (\id + D_p) := \frac12 (\id +\asrt_p - \asrt_{p^\perp})$ acts as the Jordan product operator of $p$, i.e.~that we can define $p*a := T_p a$. For an element $a=\sum_i \lambda_i p_i \in V_A$ we then define its Jordan product operator by linearity as $T_a := \sum_i \lambda_i T_{p_i}$. Each element of $V_A$ can be written as the norm limit of elements of the form $\sum_i \lambda_i p_i$, and hence we get a Jordan product for all elements by continuity.

    The crucial point we need to check is commutativity of the Jordan product. For this it suffices to check commutativity on the sharp predicates: $p*q = q*p$. This translates to $T_p q = T_q p$. Because $T_p 1 = p$ and $T_q 1 = q$ we can write our desired identity as $[T_p, T_q]1 = 0$, where $[f,g]:=f\circ g - g\circ f$ is the standard commutator bracket. The statement $[T_p, T_q]1 = 0$ is easily seen to be equivalent to $[D_p, D_q]1 = 0$. We have shown that $D_p$ and $D_q$ are order derivations (Proposition~\ref{prop:assert-is-derivation}). The commutator of two order derivations is again an order derivation~\cite[Prop.~1.114]{alfsen2012state}. These facts are combined together with some algebra in Theorem~9.48 of~\cite{alfsen2012geometry} to show that indeed $[D_p, D_q]1 = 0$.

    As mentioned above, for elements $a:=\sum_i\lambda_i p_i$ and $b:=\sum_j\mu_j q_j$ in $V_A$ we define $a*b := T_a b$ where $T_a := \sum_i\lambda_i T_{p_i}$. That this is well-defined follows from the commutativity of the $T_{p_i}$ and $T_{q_j}$. We can copy the argument of Theorem~9.43 of~\cite{alfsen2012geometry} to show it is continuous in the norm, so that it extends to all elements of $V_A$. Similarly, we can follow Theorem~9.43 of~\cite{alfsen2012geometry} to show that $*$ satisfies the Jordan identity and the final implication $-1\leq a \leq 1 \implies 0\leq a*a\leq 1$ of Definition~\ref{def:JB-algebra}.
\end{proof}

\subsection{Proof of the main theorems}

We can now prove our main reconstruction results.
First, the statement for sequential effectuses with arbitrary scalars.

\begin{proof}[Proof of Theorem~\ref{thm:JB-embedding}]
By assumption our effectus $\catC$ has filters and comprehensions and is separated by states. By Proposition~\ref{prop:sequential-has-images} it has images, and by Lemma~\ref{lem:compatible-filter-compression} the filters and comprehensions are compatible.
Hence, Proposition~\ref{prop:splits-directed-complete} applies and $\catC\cong \catC_1\times \catC_2$ where the predicate spaces of $\catC_1$ are orthoalgebras and those of $\catC_2$ are convex. By assumption the predicate spaces are normal SEAs, so that by Lemma~\ref{lem:ortho-SEA-boolean} the predicate spaces of $\catC_1$ are complete Boolean algebras. That the predicate spaces of $\catC_2$ are unit intervals of JB-algebras is given by Proposition~\ref{prop:is-JB-algebra}. The theorem now follows easily.
\end{proof}

Unfortunately, we don't know any way in which we can restrict the directed-complete JB-algebras to JBW-algebras in this setting with arbitrary scalars. 
The fact that we require all maps to be normal and that the effectus is separated by states is not enough to give 'separation by normal states' in the sense required for JBW-algebras (cf.~Definition~\ref{def:order-separating}). However, when we restrict the scalars to irreducible scalars, then the concept of state for the effectus corresponds to that for JB-algebras, and hence we do get JBW-algebras.

\begin{proof}[Proof of Theorem~\ref{thm:JBW-CBA}]
    Let $\catC$ be a sequential effectus with irreducible scalars. Then $\Pred(I)=\{0\}$, or $\Pred(I)=\{0,1\}$ or $\Pred(I)=[0,1]$. In the first case, $\catC$ is equivalent to the trivial single-object category, and hence the theorem is trivially true.
    If $\Pred(I)=\{0,1\}$, then the predicate spaces are orthoalgebras, and as they are also normal SEAs, they are complete Boolean algebras, and hence we get $\Pred\colon \catC\rightarrow \textbf{CBA}^\opp$ as required.
    Finally, if $\Pred(I)=[0,1]$ then all the predicate spaces are convex, so that the previous results give $\Pred\colon \catC\rightarrow \textbf{JB}_{\text{npc}}^\opp$.
    Let $A\in\catC$ be an object, then the predicate space $\Pred(A)$ is separated by normal states $\omega\colon I\rightarrow A$. These correspond to normal states $\omega^*\colon V_A\rightarrow \R$ on the JB-algebra $V_A$. As $V_A$ is then directed-complete and separated by normal states, it is a JBW-algebra. Hence, the predicate functor restricts to 
    $\Pred\colon \catC\rightarrow \textbf{JBW}_{\text{npc}}^\opp$.
\end{proof}

With these theorems we see that each object in a sequential effectus splits up into an object whose predicates form a Boolean algebra and an object whose predicates correspond to the unit interval of a JB-algebra, a model for a quantum system.
If the scalars of the effectus are irreducible then the predicate functor restricts to either Boolean algebras or JBW-algebras, meaning that any such category either models deterministic classical logic, or probabilistic quantum logic.

We can however still get a bit closer to quantum theory, and see that each predicate space embeds into a von Neumann algebra.

\section{Reconstruction for monoidal effectuses}\label{sec:monoidal-reconstruction}

One could argue that JBW-algebras are not a proper model for a quantum system since they also include the purely exceptional algebras (cf.~Theorems~\ref{thm:JBW-decomposition} and~\ref{thm:purely-exceptional-char}).
It is however relatively straightforward to show that if our sequential effectus is monoidal in a suitable way---namely when the pure maps and dagger are compatible with the tensor product---then the only allowed JBW-algebras for the predicate spaces are JW-algebras.
This reflects the result in finite dimension, shown in many different contexts, that a distinguishing factor between Jordan algebras and C$^*$-algebras is that the former do not allow well-behaved tensor products~\cite{barnum2016composites,wetering2018sequential,wetering2018reconstruction,selby2018reconstructing}.

\begin{definition}
    A \Define{monoidal sequential effectus} is a sequential effectus that is monoidal and such that
    \begin{itemize}
        \item The tensor product of two pure maps is pure,
        \item for pure $f$ and $g$ we have $(f\otimes g)^\dagger = f^\dagger\otimes g^\dagger$.
    \end{itemize}
\end{definition}

For this section let us assume that we are working with a monoidal sequential effectus with irreducible scalars. Then by Theorem~\ref{thm:JBW-CBA} we may assume all our predicate spaces are either Boolean algebras or JBW-algebras. We are interested in the latter case, so let us assume the scalars are $[0,1]$ instead of $\{0,1\}$ so that our predicate spaces are JBW-algebras.
In this section we will show that the monoidal structure forces our predicate spaces to be JW-algebras, which boils down to showing that the predicate spaces cannot contain exceptional subalgebras.
In this section we will need to use some more structure present in JBW-algebras than before. We will introduce the necessary concepts when needed.

For the remainder of this section let $\catC$ denote a monoidal sequential effectus with scalars $[0,1]$. Let $A$ and $B$ denote objects in $\catC$. By previous results we have JBW-algebras $V_A$ and $V_B$ such that $\Pred(A)\cong [0,1]_{V_A}$ and $\Pred(B)\cong [0,1]_{V_B}$. Additionally, the tensor product $A\otimes B$ has an associated JBW-algebra $V_{A\otimes B}$. Recall that the definition of a monoidal effectus (Definition~\ref{def:monoidal-effectus}) gives us $\truth_A\otimes \truth_B=\truth_{A\otimes B}$ and $(a\ovee b)\otimes c = (a\otimes c)\ovee (b\otimes c)$. The scalar action of $[0,1]$ on the predicates also works nicely with the tensor product (see Lemma~\ref{lem:monoidal-effectus-facts}).
As a result we get a bilinear positive unital map $V_A\times V_B\rightarrow V_{A\otimes B}$, which we will also denote by $\otimes$.

Our first step is getting a better handle on the interaction of the Jordan product with the tensor product.
First we recall from Section~\ref{sec:seqprod} that a JBW-algebra has a sequential product operation given by the quadratic product as $a\mult b = Q_{\sqrt{a}} b$. As is shown in~\cite[Theorem~4.6.17]{vandewetering2021thesis}, the map $Q_{\sqrt{a}}$ is the unique $\diamond$-positive map on a JBW-algebra satisfying $Q_{\sqrt{a}}1=a$. As $\asrt_a$ is a $\diamond$-positive map on the JBW-algebra $V_A$ satisfying $\truth \circ \asrt_a = a$ we must then have $\asrt_a Q_{\sqrt{a}}$ on $V_A$.

\begin{proposition}\label{prop:JBW-tensor-preserves-assert}
  Let $a\in \Pred(A)$ and $b\in \Pred(B)$. Then $\asrt_{a\otimes b} = \asrt_a\otimes \asrt_b$.
  In particular $\asrt_{a\otimes \truth} = \asrt_a\otimes \id$ and $a^2\otimes b^2 = (a\otimes b)^2$.
\end{proposition}
\begin{proof}
  Note $\truth\circ (\asrt_a\otimes \asrt_b) = (\truth\otimes \truth)\circ (\asrt_a\otimes \asrt_b) = (\truth\circ \asrt_a)\otimes (\truth\circ \asrt_b) = a\otimes b$, so by uniqueness of $\dagger$-positive maps, it remains to show that $\asrt_a\otimes \asrt_b$ is $\dagger$-positive. But this follows because $(\asrt_{\sqrt{a}}\otimes \asrt_{\sqrt{b}})^2 = \asrt_{\sqrt{a}}^2 \otimes \asrt_{\sqrt{b}}^2 = \asrt_a\otimes \asrt_b$.

  For $\asrt_{a\otimes \truth} = \asrt_a\otimes \id$ we simply note that $\asrt_\truth = \id$,
  and for $a^2\otimes b^2 = (a\otimes b)^2$ we calculate $a^2\otimes b^2 = (a\circ \asrt_a)\otimes (b\circ \asrt_b) = (a\otimes b)\circ (\asrt_a\otimes \asrt_b) = (a\otimes b)\circ \asrt_{a\otimes b} = (a\otimes b)^2$.
\end{proof}

\begin{corollary}\label{cor:tensor-idempotents}
  Let $p\in \Pred(A)$ and $q\in \Pred(B)$ be sharp predicates.
    Then $p\otimes q$ is sharp.
\end{corollary}

Denote by $*$ the Jordan product on $V_A$ and write $T_a(b) = a*b$ for the Jordan product map of $a$ (and similarly for $V_B$). $T_a$ is not a positive map and hence cannot be part of the effectus. However, if $p$ is sharp (\ie~idempotent) then $T_p = \frac12 (\id + Q_p - Q_{p^\perp})$ so that it is a linear combination of maps that do lie in the effectus, so that we can still speak of tensor products of these maps. Note that additionally $Q_p = \asrt_{p^2} = \asrt_p$.

\begin{proposition}
    Let $a\in V_A$ be arbitrary and $\truth\in V_B$, then $T_{a\otimes \truth} = T_a\otimes \id$. Similarly, for $\truth\in V_A$ and $b\in V_B$ we have $T_{\truth\otimes b} = \id \otimes T_b$.
\end{proposition}
\begin{proof}
    We only show the first equation, as the second follows analogously. We prove the result for $a=p$ sharp. By the norm-continuity and linearity of the Jordan product in the first argument, this is sufficient as the sharp elements span a dense set.

    Note first that $(p\otimes \truth)^\perp = p^\perp \otimes \truth$ and $\id_{A\otimes B} = \id_A\otimes \id_B$.
    We then calculate:
    \begin{align*}
    T_{p\otimes \truth} &
        \ =\  \frac12 \Bigl(\id\otimes \id + \asrt_{p\otimes \truth} - \asrt_{(p\otimes \truth)^\perp}\Bigr)
        \ =\  \frac12 \Bigl(\id \otimes \id + \asrt_p \otimes \id - \asrt_{p^\perp}\otimes \id\Bigr) \\
    &\ =\  \Bigl(\frac12 (\id +\asrt_p - \asrt_{p^\perp})\Bigr)\otimes \id\  =\  T_p \otimes \id. \qedhere
    \end{align*}
\end{proof}

The Jordan product is of course commutative, however, this does not mean that the Jordan product maps $T_a$ all commute. When $T_aT_b=T_bT_a$ we say that $a$ and $b$ \Define{operator commute}.

\begin{corollary}
    For all $a\in V_A$ and $b\in V_B$, $a\otimes \truth$ and $\truth\otimes b$ operator commute.
\end{corollary}

\begin{proposition}\label{prop:tensor-is-Jordan-hom}
    The maps $a\mapsto a\otimes \truth$ and $b\mapsto \truth\otimes b$ are normal injective Jordan homomorphisms.
\end{proposition}
\begin{proof}
    We only show this for $a\mapsto a\otimes \truth$ as the other one follows analogously.
    That it is a Jordan homomorphism, i.e.~$(a_1\otimes 1)*(a_2\otimes 1) = (a_1*a_2)\otimes 1$, follows 
    immediately from the previous corollary.

    To show it is injective suppose $a\otimes \truth = a'\otimes \truth$. Let $\omega$ be any state on the first system, and $\omega'$ any state on the second system. Then $\omega(a) = \omega(a)\omega'(\truth) = (\omega\otimes \omega')(a\otimes \truth) = (\omega\otimes \omega')(a'\otimes \truth) = \omega(a')$. Since states separate the predicates, we then necessarily have $a=a'$.

    Now as it is an injective unital Jordan homomorphism, the restriction to their domain is an order-isomorphism, and hence the maps must be normal.
\end{proof}

We now wish to prove that the quadratic product `commutes' with the tensor product for arbitrary elements of the JBW-algebras: $Q_{a\otimes b} = Q_a\otimes Q_b$ (equality here is understood as equality as linear maps on $V_{A\otimes B}$). Proposition~\ref{prop:JBW-tensor-preserves-assert} shows this result for effects $a$ and $b$. To extend this result for arbitrary $a$ and $b$ we need a concept related to the quadratic product, known as the \Define{triple product}. For $a,b,c\in V$ for $V$ a Jordan algebra we define $Q_{a,b} c = (a*b)*c + (c*b)*a - (a*c)*b$. To motivate this, when $V$ is a JW-algebra, this gives $Q_{a,b}c = \frac12(acb+bca)$ where the product is in the underlying von Neumann algebra. The triple product is related to the quadratic product via $Q_a = Q_{a,a}$.
Note that $Q_{a,b}=Q_{b,a}$ and that it is bilinear in its two arguments: $Q_{a_1+a_2,b} = Q_{a_1,b} + Q_{a_2,b}$.
In particular, $Q_{a_1+a_2} = Q_{a_1+a_2,a_1+a_2} = Q_{a_1,a_1} + Q_{a_2,a_2} + 2Q_{a_1,a_2}$.

\begin{proposition}\label{prop:tensor-quadratic}
    Let $a\in V_A$ and $b\in V_B$ be arbitrary. Then $Q_{a\otimes b} = Q_a\otimes Q_b$.
\end{proposition}
\begin{proof}
    First suppose $a\in[0,1]_{V_A}$ and $b\in[0,1]_{V_A}$. Then $Q_a = \asrt_{a^2}$ and $Q_b=\asrt_{b^2}$, so that by Proposition~\ref{prop:JBW-tensor-preserves-assert} $Q_{a\otimes b} = \asrt_{(a\otimes b)^2} = \asrt_{a\otimes b}^2 = (\asrt_a\otimes \asrt_b)^2 = \asrt_{a^2}\otimes \asrt_{b^2} = Q_a\otimes Q_b$. Now for an arbitrary positive $a$ $Q_a = Q_{\norm{a} a/\norm{a}} = \norm{a}^2 Q_{a/\norm{a}}$, and hence the desired result also follows when $a\geq 0$ and $b\geq 0$.

    Now suppose $a=a_1+a_2$ where $a_1,a_2\geq 0$. Then $Q_a = Q_{a_1} + Q_{a_2} + 2Q_{a_1,a_2}$. We expand $Q_{a\otimes b}$ in two different ways. First we see that $Q_{a\otimes b} = Q_{a_1\otimes b + a_2\otimes b} = Q_{a_1\otimes b} + Q_{a_2\otimes b} + 2 Q_{a_1\otimes b, a_2\otimes b} = Q_{a_1}\otimes Q_b + Q_{a_2}\otimes Q_b + 2 Q_{a_1\otimes b, a_2\otimes b}$.
    Secondly, $Q_{a\otimes b} = Q_a\otimes Q_b = Q_{a_1}\otimes Q_b + Q_{a_2}\otimes Q_b + 2Q_{a_1,a_2}\otimes Q_b$. Comparing terms in both of these decompositions of $Q_{a\otimes b}$ we see that necessarily $Q_{a_1\otimes b, a_2\otimes b} = Q_{a_1,a_2}\otimes Q_b$.

    We can use this equation, and do a similar trick, but starting with $Q_{a_1\otimes b, a_2\otimes b}$ where $b=b_1 + b_2$ to give us the equation $2Q_{a_1,a_2}\otimes Q_{b_1,b_2} = Q_{a_1\otimes b_1, a_2\otimes b_2} + Q_{a_1\otimes b_2, a_2\otimes b_1}$.

    Finally, suppose $a$ and $b$ are arbitrary. Write $a=a^+ - a^-$ and $b=b^+-b^-$ where $a^+,a^-, b^+, b^- \geq 0$. Now if we expand both the expression $Q_{(a^+-a^-)\otimes (b^+-b^-)}$ and $Q_{a^+-a^-} \otimes Q_{b^+ - b^-}$ as much as possible using linearity and apply the previous rewrite rules, it is easily verified that these two expression are indeed equal.
\end{proof}

In order to proceed we need to use the concept of \emph{universal von Neumann algebras}.

\begin{theorem}[{\cite[Theorem~7.1.9]{hanche1984jordan}}]
    Let $V$ be a JBW-algebra. Then there exists an (up to isomorphism) unique von Neumann algebra $W^*(V)$ and a normal Jordan homomorphism $\psi\colon V\rightarrow W^*(V)_\sa$ such that $\psi(V)$ generates $W^*(V)$ as a von Neumann algebra and if $\mathfrak{B}$ is a von Neumann algebra with a normal Jordan homomorphism $\phi\colon V\rightarrow \mathfrak{B}_\sa$, then there is a unique normal $*$-homomorphism $\hat{\phi}\colon W^*(V)\rightarrow \mathfrak{B}$ such that $\hat{\phi}\circ \psi = \phi$.
\end{theorem}

\begin{corollary}\label{cor:JW-injective}
    A JBW-algebra $V$ is a JW-algebra if and only if $\psi\colon V\rightarrow W^*(V)$ is injective.
\end{corollary}
\begin{proof}
    If $\psi$ is injective, then $V$ is of course a JW-algebra. Conversely, if $V$ is a JW-algebra, then there must be an injective normal Jordan homomorphism $\phi\colon V\rightarrow \mathfrak{B}_\sa$ for some von Neumann algebra $\mathfrak{B}$, and hence by the universal property of $W^*(V)$, $\hat{\phi}\circ \psi = \phi$, which shows that $\psi$ must be injective.
\end{proof}

\begin{definition}
    Let $V$ be a JBW-algebra. We call $s\in V$ a \Define{symmetry} when $s^2 = 1$. Two idempotents $p,q \in V$ are \Define{exchangeable by a symmetry} if there exists a symmetry $s$ such that $Q_s p = q$.
\end{definition}

\begin{lemma}[{\cite[Lemma~4.4]{alfsen2012geometry}}]\label{lem:exchangeable-by-4-is-JW}
    Let $V$ be a JBW-algebra where the identity is the sum of at least 4 idempotents that are mutually exchangeable by a symmetry. Then $V$ is a JW-algebra.
\end{lemma}

\begin{lemma}\label{lem:symmetry-exceptional}
    Let $V\neq \{0\}$ be a purely exceptional JBW-algebra. Then the identity of $V$ is the sum of 3 orthogonal non-zero idempotents exchangeable by a symmetry.
\end{lemma}
\begin{proof}
    By Theorem~\ref{thm:purely-exceptional-char} we can write $V=C(X,E)$ where $E=M_3(\mathbb{O})_{\sa}$ for some hyperstonean space $X$. As $X$ is a type I$_3$ JBW-factor there exist orthogonal non-zero idempotents $q_1,q_2,q_3\in E$ mutually exchangeable by a symmetry such that $q_1+q_2+q_3 = 1_E$~\cite[Theorem 2.8.3]{hanche1984jordan}. Let $s_{ij}\in E$ for $i,j\in\{1,2,3\}$ be symmetries so that $Q_{s_{ij}} q_i = q_j$. Define then $f_i\colon X\rightarrow E$ as the constant function $f_i(x) = q_i$, and similarly $g_{ij}\colon X\rightarrow E$ by $g_{ij}(x) = s_{ij}$. Then indeed for every $x\in X\colon (Q_{g_{ij}} f_i)(x) = Q_{s_{ij}} q_i = q_j = f_j(x)$.
\end{proof}

\begin{lemma}\label{lem:tensor-symmetry}
  Let $p_1,q_1\in V_A$ be idempotents exchangeable by a symmetry $s_1\in V_A$, and let $p_2,q_2\in V_B$ be idempotents exchangeable by a symmetry $s_2\in V_B$. Then $p_1\otimes p_2$ and $q_1\otimes q_2$ are idempotents exchangeable by $s_1\otimes s_2$.
\end{lemma}
\begin{proof}
  That $p_1\otimes p_2$ and $q_1\otimes q_2$ are idempotents follows by Corollary~\ref{cor:tensor-idempotents}. That $s_1\otimes s_2$ is a symmetry follows by Proposition~\ref{prop:tensor-quadratic}, because $(s_1\otimes s_2)^2 = Q_{s_1\otimes s_2} 1 = (Q_{s_1}\otimes Q_{s_2})(1\otimes 1) = s_1^2\otimes s_2^2 = 1\otimes 1 = 1$. By the same proposition: $Q_{s_1\otimes s_2} (p_1\otimes p_2) = (Q_{s_1}\otimes Q_{s_2})(p_1\otimes p_2) = (Q_{s_1}p_1)\otimes (Q_{s_2}p_2) = q_1\otimes q_2$.
\end{proof}

\begin{proposition}
    $V_A$ is a JW-algebra.
\end{proposition}
\begin{proof}
    Since $V_A$ is a JBW-algebra we can write $V_A = V_1\oplus V_2$ where $V_1$ is a JW-algebra and $V_2$ is purely exceptional (Theorem~\ref{thm:JBW-decomposition}). We need to show that $V_2 =\{0\}$. Towards contradiction, suppose that $V_2\neq \{0\}$.

    Let $p\in V_A$ be the central idempotent corresponding to $V_2$. Then we have an object $A_p$ in $\catC$, by considering the comprehension $\pi_p$, so that $\Pred(A_p)\cong [0,1]_{V_2}$. Let $q_1,q_2,q_3$ be a set of idempotents in $A_p$ exchangeable by symmetries $s_{ij}$ for $i,j\in \{1,2,3\}$, which exists by Lemma~\ref{lem:symmetry-exceptional}. Consider the system $A_p\otimes A_p$. By Lemma~\ref{lem:tensor-symmetry} $s_{ik}\otimes s_{jl}$ is a symmetry for all $i,k,j,l\in \{1,2,3\}$. This set of symmetries makes all nine idempotents $\{q_i\otimes q_j~;~i,j\in\{1,2,3\}\}$ in $A_p\otimes A_p$ mutually exchangeable by a symmetry.

  Hence, by Lemma~\ref{lem:exchangeable-by-4-is-JW}, $V_{A_p\otimes A_p}$ must be a JW-algebra. So then $V_{A_p\otimes A_p}$ embeds into $W^*(V_{A_p\otimes A_p})$ via an injective Jordan homomorphism (Corollary~\ref{cor:JW-injective}). 
  But we also have an injective Jordan homomorphism from $V_2$ to $V_{A_p\otimes A_p}$ given by $a\mapsto a\otimes 1$ (Proposition~\ref{prop:tensor-is-Jordan-hom}). Hence, $V_2$ embeds into $W^*(V_{A_p\otimes A_p})$. This contradicts the fact that $V_2$ is purely exceptional, so that we indeed must have had $V_2=\{0\}$.
\end{proof}

Let us denote by $\textbf{JW}_{\text{npc}}$ the full subcategory of $\textbf{JBW}_{\text{npc}}$ consisting of the JW-algebras. Combining what we have seen before we then get the following theorem.

\begin{theorem}\label{thm:JW-algebra}
  Let $\catC$ be a monoidal sequential effectus with irreducible scalars not equal to $\{0,1\}$. Then there is a functor $F\colon \catC\rightarrow \textbf{JW}_{\text{npc}}^\opp$ satisfying $F(\Pred(A))\cong [0,1]_{F(A)}$. This functor is faithful if and only if $\catC$ is separated by predicates.
\end{theorem}

Recall that the other option for the irreducible scalars is that they are equal to the Booleans $\{0,1\}$, which results in predicate spaces being complete Boolean algebras. So we either get deterministic classical systems, or probabilistic quantum systems in the form of (subspaces) of von Neumann algebras.

Note that none of our assumptions requires there to be a non-commutative algebra for one of the predicate spaces. The category could be entirely classical in the sense that all the underying von Neumann algebras are commutative. We could impose any of a number of additional assumptions that would require there to be truly quantum systems, such as the existence of dilations, which are explored in the context of effectus theory in~\cite[Section~3.7.1]{basthesis}.

Not all JW-algebras are allowed in our setting, but analysing which ones precisely has proven difficult. By adapting the arguments of~\cite{wetering2018reconstruction} we can show that no finite-dimensional quaternionic systems are allowed, but it is unclear how to adapt the argument for infinite-dimensional quaternionic systems. 
We also conjecture that no `true' spin-factors (those that aren't isomorphic to a matrix algebra) are allowed, which means that our JW-algebras restrict to the \emph{universally reversible} ones, for which a characterisation of the universal von Neumann algebra is known (cf.~Theorem~6.2.5 and Proposition~7.3.3 of~\cite{hanche1984jordan}).

\section{Conclusion}\label{sec:conclusion}
We have shown that an effectus with directed-complete predicate spaces and suitable additional structure embeds into a category of Boolean algebras and a category of JB-algebras. Requiring the scalars to be irreducible allows us to restrict to JBW-algebras, and imposing a tensor product restricts us further to von Neumann algebras.
This demonstrates that quantum theory, including both infinite-dimensional systems as well as mixed quantum-classical systems, can be reconstructed from abstract categorical grounds without even a priori referring to the structure of real numbers or convex sets.
While there have been other categorical approaches to reconstructing quantum theory~\cite{tull2016reconstruction,selby2018reconstructing}, they had to insert the real numbers at some point to get standard quantum theory, making our reconstruction the first to be fully categorical.

For our results we had to require that the predicates spaces were sequential effect algebras. This condition is not as natural as our other assumptions, and also not necessary in the finite-dimensional setting~\cite{wetering2018reconstruction}.
A natural question is therefore whether we could do without this assumption.
It should be noted that the assumptions in Definition~\ref{def:sequential-effectus} are not minimal, and some of the axioms of a sequential product follow `for free'. The ones that require explicit inclusion are related to the commutativity of certain assert maps. Our assumptions could be improved by finding a way to derive these conditions in a more categorically natural way.
However, in~\cite{wetering2019commutativity} these commutativity conditions are derived using the Fuglede--Putnam--Rosenblum theorem, so it is unlikely that such a categorical characterisation would be straightforward if it would imply the same results.

Another open question is how we can get get the systems to correspond to JBW-algebras without restricting the scalars to be irreducible.
Additionally it would be interesting to characterise which subcategories of JW-algebras satisfy our assumptions, and in particular which JW-algebras are not allowed by our assumptions.

\paragraph{Acknowledgments} JvdW is funded by an NWO Rubicon personal fellowship. The authors would like to thank the anonymous LiCS reviewer that found an oversight in an earlier version of our proof, which we fixed by including Proposition~\ref{prop:assert-is-derivation}.

\bibliographystyle{plainnat}
\bibliography{main}

\end{document}